\documentclass[12pt]{article}

\usepackage[colorlinks, linkcolor=blue, citecolor=blue]{hyperref}           
\usepackage{color}
\usepackage{graphicx,amsmath,amssymb,amsfonts,bm,epsfig,epsf,url,dsfont,mathrsfs,enumitem}

\usepackage{amsthm}
\usepackage{tikz}
\usepackage{bbm}      
\usepackage{caption}
\usepackage[numbers]{natbib}
\usepackage{subcaption}
\usepackage[top=1in,bottom=1in,left=1in,right=1in]{geometry}
\usepackage{fancybox}
\usepackage{hyperref}
\usepackage{algorithm}
\usepackage{verbatim}
\usepackage{algorithmic}
\usepackage{mathtools}

\makeatletter
\renewcommand{\paragraph}{%
  \@startsection{paragraph}{4}%
  {\z@}{2.25ex \@plus 1ex \@minus .2ex}{-1em}%
  {\normalfont\normalsize\bfseries}%
}
\makeatother

\usepackage{scalerel,stackengine}
\stackMath
\newcommand\reallywidehat[1]{%
\savestack{\tmpbox}{\stretchto{%
  \scaleto{%
    \scalerel*[\widthof{\ensuremath{#1}}]{\kern-.6pt\bigwedge\kern-.6pt}%
    {\rule[-\textheight/2]{1ex}{\textheight}}
  }{\textheight}%
}{0.5ex}}%
\stackon[1pt]{#1}{\tmpbox}%
}

\newcommand{\indep}{\perp \!\!\! \perp}

\newcommand{\bE}{\mathbb{E}}

\newcommand{\KL}{D_{\mathrm{KL}}}

\newtheorem{definition}{Definition}

\newtheorem{theorem}{Theorem}
\newtheorem{corollary}[theorem]{Corollary}
\newtheorem{conjecture}{Conjecture}
\newtheorem{lemma}{Lemma}
\newtheorem{prop}{Proposition}

\newtheorem{rmk}{Remark}

\newenvironment{fminipage}%
  {\begin{Sbox}\begin{minipage}}%
  {\end{minipage}\end{Sbox}\fbox{\TheSbox}}

\newcommand*{\rom}[1]{\expandafter\@slowromancap\romannumeral #1@}

\newcommand{\Ind}{\mathbbm{1}}

\newcommand{\yb}{\mathbf{y}}
\newcommand{\Ib}{\mathbf{I}}
\newcommand{\Jb}{\mathbf{J}}

\newcommand{\abs}[1]{\left|#1\right|}

\newcommand{\cN}{\mathcal{N}}

\newcommand {\pr} {\mathbb{P}}

\newcommand{\calA}{{\cal A}}
\newcommand{\calB}{{\cal B}}
\newcommand{\calC}{{\cal C}}
\newcommand{\calD}{{\cal D}}
\newcommand{\calE}{{\cal E}}

\newcommand{\calH}{{\cal H}}

\newcommand{\calK}{{\cal K}}

\newcommand{\calN}{{\cal N}}

\newcommand{\calP}{{\cal P}}
\newcommand{\calQ}{{\cal Q}}

\newcommand{\calS}{{\cal S}}

\newcommand{\calU}{{\cal U}}
\newcommand{\calV}{{\cal V}}

\newcommand{\be}{\begin{equation}}
\newcommand{\ee}{\end{equation}}
\newcommand{\beqna}{\begin{eqnarray}}
\newcommand{\eeqna}{\end{eqnarray}}

\newcommand{\p}[1]{\left(#1\right)}
\newcommand{\pp}[1]{\left[#1\right]}
\newcommand{\ppp}[1]{\left\{#1\right\}}
\newcommand{\norm}[1]{\left\|#1\right\|}
\newcommand{\innerP}[1]{\left\langle#1\right\rangle}

\makeatletter
\renewcommand{\paragraph}{%
  \@startsection{paragraph}{4}%
  {\z@}{1ex \@plus 1ex \@minus .2ex}{-1em}%
  {\normalfont\normalsize\bfseries}%
}

\makeatletter
\def\thanks#1{\protected@xdef\@thanks{\@thanks
        \protect\footnotetext{#1}}}
\makeatother

\newcommand{\s}[1]{\mathsf{#1}}
\usepackage{xspace}

\DeclareMathOperator*{\argmax}{arg\,max}

\allowdisplaybreaks

\begin{document}

\title{Detection and Recovery of Hidden Submatrices}


\author{Marom Dadon~~~~Wasim Huleihel~~~~Tamir Bendory\thanks{M. D., W. H., and T. B. are with the Department of Electrical Engineering-Systems at Tel Aviv university, {T}el {A}viv 6997801, Israel (e-mails:  \texttt{marom.dadon@gmail.com, wasimh@tauex.tau.ac.il, bendory@tauex.tau.ac.il }). W.H. is supported by ISF under Grant 1734/21. T.B. is supported in part by BSF under Grant 2020159, in part by NSF-BSF under Grant 2019752, and in part by ISF under Grant 1924/21.}}


\maketitle

\begin{abstract}
In this paper, we study the problems of detection and recovery of hidden submatrices with elevated means inside a large Gaussian random matrix. We consider two different structures for the planted submatrices. In the first model, the planted matrices are disjoint, and their row and column indices can be arbitrary. Inspired by scientific applications, the second model restricts the row and column indices to be consecutive. In the detection problem, under the null hypothesis, the observed matrix is a realization of independent and identically distributed standard normal entries. Under the alternative, there exists a set of hidden submatrices with elevated means inside the same standard normal matrix. Recovery refers to the task of locating the hidden submatrices. For both problems, and for both models, we characterize the statistical and computational barriers by deriving information-theoretic lower bounds, designing and analyzing algorithms matching those bounds, and proving computational lower bounds based on the low-degree polynomials conjecture. In particular, we show that the space of the model parameters (i.e., number of planted submatrices, their dimensions, and elevated mean) can be partitioned into three regions: the \emph{impossible} regime, where all algorithms fail; the \emph{hard} regime, where while detection or recovery are statistically possible, we give some evidence that polynomial-time algorithm do not exist; and finally the \emph{easy} regime, where polynomial-time algorithms exist. 
\end{abstract}

\section{Introduction}\label{sec:intro}

This paper studies the detection and recovery problems of hidden submatrices inside a large Gaussian random matrix. In the \emph{detection problem}, under the null hypothesis, the observed matrix is a realization of an independent and identically distributed random matrix with standard normal entries. Under the alternative, there exists a set of hidden submatrices with elevated means inside the same standard normal matrix. Our task is to design a statistical test (i.e., an algorithm) to decide which hypothesis is correct. The \emph{recovery task} is the problem of locating the hidden submatrices. In this case, the devised algorithm estimates the location of the submatrices. 

We consider two statistical models for the planted submatrices. In the first model, the planted matrices are disjoint, and their row and column indices can be arbitrary. The detection and recovery variants of this model are well-known as the \emph{submatrix detection} and \emph{submatrix recovery} (or localization) problems, respectively, and received significant attention in the last few years, e.g.,  \cite{shabalin2009finding,kolar2011minimax,balakrishnan2011statistical,butucea2013detection,arias2014community,hajek2015computational,montanari2015limitation,verzelen2015community,ma2015computational,XingNobel,Arias10,Bhamidi17,chen2016statistical,cai2015computational,brennan18a,brennan19,9917525}, and  references therein. Specifically, for the case of a \emph{single} planted submatrix, the task is to detect the presence of a small $k\times k$ submatrix with entries sampled from a distribution~$\calP$ in an $n\times n$ matrix of samples from a distribution $\calQ$. In the special case where $\calP$ and $\calQ$ are Gaussians, the statistical and computational barriers, i.e., information-theoretic lower bounds, algorithms, and computational lower bounds, were studied in great detail and were characterized in~\cite{butucea2013detection,montanari2015limitation,shabalin2009finding,kolar2011minimax,balakrishnan2011statistical,ma2015computational,brennan19}. When $\calP$ and $\calQ$ are Bernoulli random variables, the detection task is well-known as the planted dense subgraph problem, which has also been studied extensively in the literature, e.g., \cite{butucea2013detection,arias2014community,verzelen2015community,hajek2015computational,brennan18a}. Most notably, for both the Gaussian and Bernoulli problems, it is well understood by now that there appears to be a statistical-computational gap between the minimum value of $k$ at which detection can be solved, and the minimum value of $k$ at which detection can be solved in polynomial time (i.e., with an efficient algorithm). The statistical and computational barriers to the recovery problem have also received significant attention in the literature, e.g., \cite{chen2016statistical,montanari2015finding,candogan2018finding,hajek2016achieving,hajek2016information,cai2015computational,brennan18a}, covering several types of distributions, as well as single and (non-overlapping) multiple planted submatrices.

The submatrix model above, where the planted column and row indices are arbitrary, might be less realistic in certain scientific and engineering applications. Accordingly, we also analyze a second model that restricts the row and column indices to be consecutive. One important motivation for this model stems from single-particle cryo-electron microscopy (cryo-EM): a leading technology to elucidate the three-dimensional atomic structure of macromolecules, such as proteins~\cite{bai2015cryo,lyumkis2019challenges}. At the beginning of the algorithmic pipeline of cryo-EM, it is required to locate multiple particle images (tomographic projections of randomly oriented copies of the sought molecular structure) in a highly noisy, large image~\cite{singer2018mathematics,bendory2020single}. This task is dubbed particle picking. While many particle picking algorithms were designed, e.g., ~\cite{wang2016deeppicker,heimowitz2018apple,bepler2019positive,eldar2020klt}, this work can be seen as a first attempt to unveil the statistical and computational properties of this task that were not analyzed heretofore. 

\paragraph{Main contributions.} To present our results, let us introduce a few notations. In our models, we have $m$ disjoint $k\times k$ submatrices planted in an $n\times n$ matrix. We denote the observed matrix by $\s{X}$. As mentioned above, we deal with the Gaussian setting, where the entries of the planted submatrices are independent Gaussian random variables with mean $\lambda>0$ and unit variance, while the entries of the other entries in $\s{X}$ are independent Gaussian random variables with zero mean and unit variance. This falls under the general ``signal+noise" model, in the sense that $\s{X} = \lambda\cdot\s{S}+\s{Z}$, with $\s{S}$ being the signal of interest, $\s{Z}$ is a Gaussian noise matrix, and $\lambda$ describes the signal-to-noise ratio (SNR) of the problem. As mentioned above, in this paper, we consider two models for $\s{S}$; the first with the arbitrary placement of the $m$ planted submatrices, and the second with each of the $m$ planted submatrices having consecutive row and column indices. We will refer to the detection/recovery of the first model as \emph{submatrix detection/recovery}, while for the second as \emph{consecutive submatrix detection/recovery}. Contrary to the consecutive submatrix detection and recovery problems, which were not studied in the literature, the non-consecutive submatrix detection and recovery problems received significant attention; our contribution in this paper to this problem is the analysis of the detection of multiple (possibly growing) number of planted submatrices, which seems to be overlooked in the literature. As mentioned above, the recovery counterpart of multiple planted submatrices was studied in, e.g., \cite{cai2015computational}. 

For the submatrix detection, the consecutive submatrix detection, and  the consecutive submatrix recovery problems, we study the computational and statistical boundaries and derive information-theoretic lower bounds, algorithmic upper bounds, and computational lower bounds. In particular, we show that the space of the model parameters $(k,m,\lambda)$ can be partitioned into different disjoint regions: the \emph{impossible} regime, where all algorithms fail; the \emph{hard} regime, where while detection or recovery are statistically possible, we give some evidence that polynomial-time algorithms do not exist; and finally the \emph{easy} regime, where polynomial-time algorithms exist. Table~\ref{tab:results} summarizes the statistical and computational thresholds for the detection and recovery problems discussed above. We emphasize that the bounds in the second row of Table~\ref{tab:results} (submatrix recovery), as well as the first row (submatrix detection) for $m=1$, are known in the literature, as mentioned above.

\begin{table}[h!]
\centering
\begin{tabular}{|p{1.1cm}||c||c||c||c|}
 \hline
\textbf{Type}& \textbf{Impossible}& \textbf{Hard} & \textbf{Easy}\\
 \hline
 $\s{SD}$ & $\lambda \ll\frac{n}{mk^2}\wedge \frac{1}{\sqrt{k}}$ & $\frac{n}{mk^2}\wedge \frac{1}{\sqrt{k}}\ll\lambda\ll1\wedge\frac{n}{mk^2}$ & $\lambda\gg1\wedge\frac{n}{mk^2}$\\
\hline
 $\s{SR}$ & $\lambda \ll\frac{1}{\sqrt{k}}$ &$\frac{1}{\sqrt{k}}\ll\lambda\ll1\wedge\frac{\sqrt{n}}{k}$ & $\lambda\gg1\wedge\frac{\sqrt{n}}{k}$\\
\hline
 $\s{CSD}$ & $\lambda\ll\frac{1}{k}$ & $\s{NO}$& $\lambda\gg\frac{1}{k}$\\
 \hline
  $\s{CSR}$ & $\lambda\ll\frac{1}{\sqrt{k}}$ & $\s{NO}$& $\lambda \gg\frac{1}{\sqrt{k}}$\\
 \hline
\end{tabular}
\caption{Statistical and computational thresholds for submatrix detection ($\s{SD}$), submatrix recovery ($\s{SR}$), consecutive submatrix detection ($\s{CSD}$), and consecutive submatrix recovery ($\s{CSR}$), up to poly-log factors. The bounds in the first row for the special case of $m=1$ and the second row, are known in the literature (e.g., \cite{butucea2013detection,ma2015computational,chen2016statistical,cai2015computational}).}
\label{tab:results}
\end{table}

Interestingly, while it is well-known that the number of planted submatrices $m$ does not play any significant role in the statistical and computational barriers in the submatrix recovery problem, it can be seen that this is not the case for the submatrix detection problem. Similarly to the submatrix recovery problem (and of course the single planted submatrix detection problem), the submatrix detection problem undergoes a statistical-computational gap. To provide evidence for this phenomenon, we follow a recent line of work \cite{hopkins2017bayesian,Hopkins18,Kunisky19,Cherapanamjeri20,gamarnik2020lowdegree} and show that the class of low-degree polynomials fail to solve the detection problem in this conjecturally hard regime. Furthermore, it can be seen that the consecutive submatrix detection and recovery problems are either impossible or easy to solve, namely, there is no hard regime. Here, for both the detection and recovery problems, the number of planted submatrices $m$ does not play an inherent role. We note that there is a statistical gap between consecutive detection and recovery; the former is statistically easier. This is true as long as exact recovery is the performance criterion. We also analyze the correlated recovery (also known as weak recovery) criterion, where recovery is successful if only a fraction of planted entries are recovered. For this weaker criterion, we show that recovery and detection are asymptotically equivalent. 

\paragraph{Notation.} 
Given a distribution $\mathbb{P}$, let $\mathbb{P}^{\otimes n}$ denote the distribution of the $n$-dimensional random vector $(X_1, X_2, \dots, X_n)$, where the $X_i$ are i.i.d.\ according to $\mathbb{P}$. Similarly, $\mathbb{P}^{\otimes m \times n}$ denotes the distribution on $\mathbb{R}^{m \times n}$ with i.i.d.\ entries distributed as $\mathbb{P}$. Given a finite or measurable set $\mathcal{X}$, let $\text{Unif}[\mathcal{X}]$ denote the uniform distribution on $\mathcal{X}$. The relation $X\indep Y$ means that the random variables $X$ and $Y$ are statistically independent. The Hadamard and inner product between two $n\times n$ matrices $\s{A}$ and $\s{B}$ are  denoted, respectively, by $\s{A}\odot\s{B}$ and $\innerP{\s{A},\s{B}}=\s{trace}(\s{A}^T\s{B})$. For $x\in\mathbb{R}$, we define $[x]_+ = \max(x,0)$. The nuclear norm of a symmetric matrix $\s{A}$ is denoted by $\norm{\s{A}}_{\star}$, and equals the summation of the absolute values of the eigenvalues of $\s{A}$.

Let $\calN(\mu, \sigma^2)$ denote a normal random variable with mean $\mu$ and variance $\sigma^2$, when $\mu \in \mathbb{R}$ and $\sigma \in \mathbb{R}_{\ge 0}$. Let $\calN(\mu, \Sigma)$ denote a multivariate normal random vector with mean $\mu \in \mathbb{R}^d$ and covariance matrix $\Sigma$, where $\Sigma$ is a $d \times d$ positive semidefinite matrix. Let~$\Phi$ denote the cumulative distribution of a standard normal random variable with $\Phi(x) = \int_{-\infty}^x e^{-t^2/2} dt$. For probability measures $\mathbb{P}$ and $\mathbb{Q}$, let $d_{\s{TV}}(\mathbb{P},\mathbb{Q})=\frac{1}{2}\int |\mathrm{d}\mathbb{P}-\mathrm{d}\mathbb{Q}|$, $\chi^2(\mathbb{P}||\mathbb{Q}) = \int\frac{(\mathrm{d}\mathbb{P}-\mathrm{d}\mathbb{Q})^2}{\mathrm{d}\mathbb{Q}}$, and $d_{\s{KL}}(\mathbb{P}||\mathbb{Q}) = \bE_{\mathbb{P}}\log\frac{\mathrm{d}\mathbb{P}}{\mathrm{d}\mathbb{Q}}$, denote the total variation distance, the $\chi^2$-divergence, and the Kullback-Leibler (KL) divergence, respectively. Let $\s{Bern}(p)$ and $\s{Binomial}(n,p)$ denote the Bernoulli and Binomial distributions with parameters $p$ and $n$, respectively. We denote by $\s{Hypergeometric}(n,k,m)$ the Hypergeometric distribution with parameters $(n,k,m)$. 

We use standard asymptotic notation. For two positive sequences $\{a_n\}$ and $\{b_n\}$, we write $a_n = O(b_n)$ if $a_n\leq Cb_n$, for some absolute constant $C$ and for all $n$; $a_n = \Omega(b_n)$, if $b_n = O(a_n)$; $a_n = \Theta(b_n)$, if $a_n = O(b_n)$ and $a_n = \Omega(b_n)$, $a_n = o(b_n)$ or $b_n = \omega(a_n)$, if $a_n/b_n\to0$, as $n\to\infty$. Finally, for $a,b\in\mathbb{R}$, we let $a\vee b\triangleq\max\{a,b\}$ and $a\wedge b\triangleq\min\{a,b\}$. Throughout the paper, $\s{C}$ refers to any constant independent of the parameters of the problem at hand and will be reused for different constants. The notation $\ll$ refers to polynomially less than in $n$, namely, $a_n\ll b_n$  if $\liminf_{n\to\infty}\log_n a_n<\liminf_{n\to\infty}\log_n b_n$, e.g., $n\ll n^2$, but $n\not\ll n\log_2 n$. For $n\in\mathbb{N}$, we let $[n] = \{1, 2, \dots, n\}$. For a subset $S \subseteq \mathbb{R}$, we let $\mathbbm{1}\{S\}$ denote the indicator function of the set~$S$.

\section{Problem Formulation}\label{sec:model}

In this section, we present our model and define the detection and recovery problems we investigate, starting with the former. For simplicity of notations, we denote $\calQ = \calN(0,1)$ and $\calP = \calN(\lambda,1)$, for some $\lambda>0$, which can be interpreted as the signal-to-noise ratio (SNR) parameter of the underlying model.

\subsection{The detection problem} Let $(m,k,n)$ be three natural numbers, satisfying $m\cdot k\leq n$. We emphasize that the values of $m$, $k$, and $\lambda$, are allowed to be functions of $n$---the dimension of the observation. Let $\calK_{k,m,n}$ denote all possible sets that can be represented as a union of $m$ disjoint subsets of $[n]$, each of size $k$; see Figure~\ref{fig:illustationEnsembles} for an illustration. Formally, 
\begin{align} 
        \calK_{k,m,n}\triangleq \biggl\{ &\s{K}_{k,m} = \bigcup_{i=1}^m\s{S}_i \times \s{T}_i:\;\s{S}_i, \s{T}_i\subset\calC_k,\;\forall i\in[m],\nonumber\\
        &\hspace{3.65cm}(\s{S}_i\times\s{T}_i) \cap (\s{S}_j\times \s{T}_j)=\emptyset,\;\forall i\neq j\in[m]\biggl\}\label{eqn:SetSD},
\end{align}
where $\calC_k\triangleq\ppp{\s{S}\subset[n]:\;|\s{S}|=k}$, namely, it is the set of all subsets of $[n]$ of size $k$.
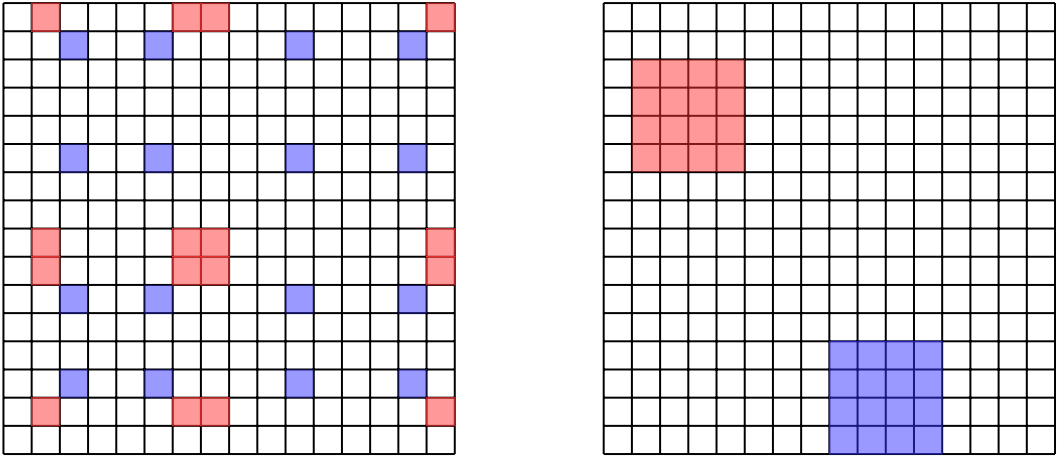
\begin{figure}
     \hspace{1.2cm}\begin{subfigure}[b]{0.3\textwidth}
     \centering
         \begin{tikzpicture}[scale=0.75]
        \draw[step=5mm, line width=0.75] (0,0) grid (8,8);
        \filldraw [semitransparent, fill=blue!80] (1,1) rectangle (1.5,1.5);
        \filldraw [semitransparent, fill=blue!80] (1,2.5) rectangle (1.5,3);
        \filldraw [semitransparent, fill=blue!80] (1,5) rectangle (1.5,5.5);
        \filldraw [semitransparent, fill=blue!80] (1,7) rectangle (1.5,7.5);
        \filldraw [semitransparent, fill=blue!80] (2.5,1) rectangle (3,1.5);
        \filldraw [semitransparent, fill=blue!80] (2.5,2.5) rectangle (3,3);
        \filldraw [semitransparent, fill=blue!80] (2.5,5) rectangle (3,5.5);
        \filldraw [semitransparent, fill=blue!80] (2.5,7) rectangle (3,7.5);
        \filldraw [semitransparent, fill=blue!80] (5,1) rectangle (5.5,1.5);
        \filldraw [semitransparent, fill=blue!80] (5,2.5) rectangle (5.5,3);
        \filldraw [semitransparent, fill=blue!80] (5,5) rectangle (5.5,5.5);
        \filldraw [semitransparent, fill=blue!80] (5,7) rectangle (5.5,7.5);
        \filldraw [semitransparent, fill=blue!80] (7,1) rectangle (7.5,1.5);
        \filldraw [semitransparent, fill=blue!80] (7,2.5) rectangle (7.5,3);
        \filldraw [semitransparent, fill=blue!80] (7,5) rectangle (7.5,5.5);
        \filldraw [semitransparent, fill=blue!80] (7,7) rectangle (7.5,7.5);
        
        \filldraw [semitransparent, red!80] (0.5,0.5) rectangle (1,1);
        \filldraw [semitransparent, red!80] (0.5,3) rectangle (1,3.5);
        \filldraw [semitransparent, red!80] (0.5,3.5) rectangle (1,4);
        \filldraw [semitransparent, red!80] (0.5,7.5) rectangle (1,8);
        \filldraw [semitransparent, red!80] (3,0.5) rectangle (3.5,1);
        \filldraw [semitransparent, red!80] (3,3) rectangle (3.5,3.5);
        \filldraw [semitransparent, red!80] (3,3.5) rectangle (3.5,4);
        \filldraw [semitransparent, red!80] (3,7.5) rectangle (3.5,8);
        \filldraw [semitransparent, red!80] (3.5,0.5) rectangle (4,1);
        \filldraw [semitransparent, red!80] (3.5,3) rectangle (4,3.5);
        \filldraw [semitransparent, red!80] (3.5,3.5) rectangle (4,4);
        \filldraw [semitransparent, red!80] (3.5,7.5) rectangle (4,8);
        \filldraw [semitransparent, red!80] (7.5,0.5) rectangle (8,1);
        \filldraw [semitransparent, red!80] (7.5,3) rectangle (8,3.5);
        \filldraw [semitransparent, red!80] (7.5,3.5) rectangle (8,4);
        \filldraw [semitransparent, red!80] (7.5,7.5) rectangle (8,8);
    \end{tikzpicture}
     \end{subfigure}
\quad\quad\quad\quad\quad\quad\quad
\begin{subfigure}[b]{0.3\textwidth}
         \centering
         \begin{tikzpicture}[scale=0.75]
    \draw[step=5mm, line width=0.75] (0,0) grid (8,8);
        \filldraw [semitransparent, fill=red!80] (0.5,5) rectangle (2.5,7);
        \filldraw [semitransparent, blue!80] (4,0) rectangle (6, 2);
    
    \end{tikzpicture}
     \end{subfigure}
        \caption{Illustration of the models considered in this paper:  $\mathcal{K}_{k,m,n}$ of Definition~\ref{eqn:SetSD} (left) and $\mathcal{K}_{k, m, n}^{\s{con}}$ 
 of Definition~\ref{def:gcons} (right), for $k=4$, $m=2$, and $n=16$.}
        \label{fig:illustationEnsembles}
\end{figure}
We next formulate two different detection problems that we wish to investigate, starting with the following one, a generalization of the Gaussian planted clique problem (or, bi-clustering, see, e.g., \cite{ma2015computational}) to multiple hidden submatrices (or, clusters). 

\begin{definition}[Submatrix detection]\label{def:SD}
Let $(\calP,\calQ)$ be a pair of distributions  over a measurable space $(\mathbb{R},\mathcal{B})$. Let $\s{SD}(n,k,m,\calP,\calQ)$ denote the hypothesis testing problem with observation $\s{X} \in \mathbb{R}^{n \times n}$ and hypotheses
\begin{align}
\calH_0: \s{X} \sim \calQ^{\otimes n \times n}\quad\quad\s{vs.}\quad\quad\calH_1: \s{X}\sim\mathcal{D}(n,k,m,\calP, \calQ),\label{eqn:modelSD}
\end{align}
where $\mathcal{D}(n, k, m,\calP, \calQ)$ is the distribution of matrices $\s{X}$ with entries $\s{X}_{ij} \sim \calP$ if $i, j \in \s{K}_{k,m}$ and $\s{X}_{ij} \sim \calQ$ otherwise that are conditionally independent given $\s{K}_{k,m}$, which is chosen uniformly at random over all subsets of $\calK_{k,m,n}$.
\end{definition}
To wit, under $\calH_0$ the elements of $\s{X}$ are all distributed i.i.d. according to $\calQ$, while under $\calH_1$, there are $m$ planted disjoint submatrices $\s{K}_{k,m}$ in $\s{X}$ with entries distributed according to~$\calP$, and the other entries (outside of  $\s{K}_{k,m}$) are distributed according to~$\calQ$. 

Note that the columns and row indices of the planted submatrices in \eqref{eqn:SetSD} can appear everywhere; in particular, they are not necessarily consecutive. In some applications, however, we would like those submatrices to be defined by a set of consecutive rows and a set of consecutive columns (e.g., when those submatrices model images like in cryo-EM). Accordingly, we consider the following set:
\begin{align} 
        \calK_{k,m,n}^{\s{con}}\triangleq \biggl\{ &\s{K}_{k,m} = \bigcup_{i=1}^m\s{S}_i \times \s{T}_i:\;\s{S}_i, \s{T}_i\subset\calC_k^{\s{con}},\;\forall i\in[m],\nonumber\\
        &\hspace{3.65cm}(\s{S}_i\times\s{T}_i) \cap (\s{S}_j\times \s{T}_j)=\emptyset,\;\forall i\neq j\in[m]\biggl\}\label{generalK},
\end{align}
where $\calC_k^{\s{con}}\triangleq\ppp{\s{S}\subset[n]:\;|\s{S}|=k,\;\s{S} \text{ is consecutive}}$, namely, it is the set of all subsets of $[n]$ of size $k$  with consecutive elements. For example, for $n=4$, we have $\calC_3^{\s{con}} = \{1,2,3\}\cup\{2,3,4\}$. The difference between $\calK_{k,m,n}$ and $\calK_{k,m,n}^{\text{con}}$ is depicted in Figure~\ref{fig:illustationEnsembles}; it is evident that the submatrices in $\calK_{k,m,n}$ can appear everywhere, while those in $\calK_{k,m,n}^{\s{con}}$ are consecutive. Consider the following detection problem. 
\begin{definition}[Consecutive submatrix detection]\label{def:gcons}
   Let $(\calP,\calQ)$ be a pair of distributions  over a measurable space $(\mathbb{R},\mathcal{B})$. Let $\s{CSD}(n,k,m,\calP,\calQ)$ denote the hypothesis testing problem with observation $\s{X} \in \mathbb{R}^{n \times n}$ and hypotheses
    \begin{align}
    \calH_0: \s{X} \sim \calQ^{\otimes n \times n}\quad\quad\s{vs.}\quad\quad\calH_1: \s{X}\sim\widetilde{\mathcal{D}}(n,k,m,\calP, \calQ),\label{eqn:modelgcons}
    \end{align}
    where $\widetilde{\mathcal{D}}(n, k, m, \calP, \calQ)$ is the distribution of matrices $\s{X}$ with entries $\s{X}_{ij} \sim \calP$ if $i, j \in \s{K}_{k,m}$ and $\s{X}_{ij} \sim \calQ$ otherwise that are conditionally independent given $\s{K}_{k,m}$, which is chosen uniformly at random over all subsets of $\calK_{k,m,n}^{\s{con}}$.
\end{definition}  

Observing $\s{X}$, a detection algorithm $\calA_n$ for the problems above is tasked with outputting a decision in $\{0,1\}$. We define the \emph{risk} of a detection algorithm $\calA_n$ as the sum of its $\s{Type}$-$\s{I}$ and $\s{Type}$-$\s{II}$ errors probabilities, namely,
\begin{align}
\s{R}(\calA_n) = \pr_{\calH_0}(\calA_n(\s{X})=1)+\pr_{\calH_1}(\calA_n(\s{X})=0),
\end{align}
where $\pr_{\calH_0}$ and $\pr_{\calH_1}$ denote the probability distributions under the null hypothesis and the alternative hypothesis, respectively. If $\s{R}(\calA_n)\to0$ as $n\to\infty$, then we say that $\calA_n$ solves the detection problem. The algorithms we consider here are either unconstrained (and thus might be computationally expensive) or run in polynomial time (computationally efficient). Typically, unconstrained algorithms are considered in order to show that information-theoretic lower bounds are asymptotically tight. An algorithm that runs in polynomial time must run in $\s{poly}(n)$ time, where~$n$ is the size of the input. As mentioned in the introduction, our goal is to derive necessary and sufficient conditions for when it is impossible and possible to detect the underlying submatrices, with and without computational constraints, for both the $\s{SD}$ and $\s{CSD}$ models.

\subsection{The recovery problem}
Next, we consider the recovery variant of the problem in Definition~\ref{def:gcons}. Note that the submatrix recovery problem that corresponds to the problem in Definition~\ref{def:SD}, where the entries of the submatrices are not necessarily consecutive, was investigated in~\cite{chen2016statistical}. In the recovery problem, we assume that the data follow the distribution under $\calH_1$ in Definition~\ref{def:gcons}, and the inference task is to recover the location of the planted submatrices. This is the analog of the particle picking problem in cryo-EM that was introduced in Section~\ref{sec:intro}.  Consider the following definition.
\sloppy

\begin{definition}[Consecutive submatrix recovery]\label{def:gconsrec}
   Let $(\calP,\calQ)$ be a pair of distributions  over a measurable space $(\mathbb{R},\mathcal{B})$. Assume that $\s{X} \in \mathbb{R}^{n \times n}\sim\widetilde{\mathcal{D}}(n,k,m,\calP, \calQ)$, where $\widetilde{\mathcal{D}}(n, k,m, \calP, \calQ)$ is the distribution of matrices $\s{X}$ with entries $\s{X}_{ij} \sim \calP$ if $i, j \in \s{K}^\star$ and $\s{X}_{ij} \sim \calQ$ otherwise that are conditionally independent given $\s{K}^\star\in\calK_{k,m,n}^{\s{con}}$. The goal is to recover the hidden submatrices $\s{K}^\star$, up to a permutation of the submatrices indices, given the matrix $\s{X}$. We let $\s{CSR}(n,k,m,\calP,\calQ)$ denote this recovery problem.
\end{definition}

Several metrics of reconstruction accuracy are possible, and we will focus on two: \emph{exact} and \emph{correlated} recovery criteria. Our estimation procedures produce a set $\hat{\s{K}} = \hat{\s{K}}(\s{X})$ aimed to estimate at best the underlying true submatrices $\s{K}^\star$. Consider the following definitions. 

\begin{definition}[Exact recovery]
    We say that $\hat{\s{K}}$ achieves exact recovery of $\s{K}^\star$, if, as $n\to\infty$, $\sup_{\s{K}^\star\in\calK_{k,m,n}^{\s{con}}}\pr(\hat{\s{K}}\neq\s{K}^\star)\to0$.
\end{definition}
\begin{definition}[Correlated recovery]
    The overlap of $\s{K}^\star$ and $\hat{\s{K}}$ is defined as the expected size of their intersection, i.e.,
    \begin{align}
        \s{overlap}(\s{K}^\star,\hat{\s{K}}) \triangleq \bE\langle{\s{K}^\star,\hat{\s{K}}}\rangle = \sum_{i=1}^n\pr(i\in\s{K}^\star\cap\hat{\s{K}}).
    \end{align}
    We say that $\hat{\s{K}}$ achieves correlated recovery of $\s{K}^\star$ if there exists a fixed constant $\epsilon>0$, such that $\lim_{n\to\infty}\sup_{\s{K}^\star\in\calK_{k,m,n}^{\s{con}}}\frac{\s{overlap}(\s{K}^\star,\hat{\s{K}})}{mk^2}\geq\epsilon$.
\end{definition}
Similarly to the detection problem, also here we will care about both unconstrained and polynomial time algorithms, and we aim to derive necessary and sufficient conditions for when it is impossible and possible to recover the underlying submatrices.

\section{Main Results}\label{sec:main}

In this section, we present our main results for the detection and recovery problems, starting with the former. For both problems, we derive the statistical and computational bounds for the two models we presented in the previous section.

\subsection{The detection problem}

\paragraph{Upper bounds.} We start by presenting our upper bounds. To that end, we propose three algorithms and analyze their performance. Define the statistics,
\begin{align}
    \s{T}_{\s{sum}}(\s{X})&\triangleq\sum_{i,j\in[n]}\s{X}_{ij},\label{eqn:secantestSum}\\
    \s{T}^{\s{SD}}_{\s{scan}}(\s{X})&\triangleq\max_{\s{K}\in\calK_{k,1,n}}\sum_{i,j\in \s{K}}\s{X}_{ij},\label{eqn:secantest0}\\
    \s{T}^{\s{CSD}}_{\s{scan}}(\s{X})&\triangleq\max_{\s{K}\in\calK_{k,1,n}^{\s{con}}}\sum_{i,j\in \s{K}}\s{X}_{ij}.\label{eqn:secantest0CSD}
\end{align}
The statistics in \eqref{eqn:secantestSum} amounts to adding up all the elements of $\s{X}$, while \eqref{eqn:secantest0} and \eqref{eqn:secantest0CSD} enumerate all $k\times k$ submatrices of $\s{X}$ in $\calK_{k,1,n}$ and $\calK_{k,1,n}^{\s{con}}$, and take the submatrix with  the maximal sum of entries, respectively. Fix $\delta>0$. Then, our tests are defined as,
\begin{align}
    \calA_{\s{sum}}(\s{X})&\triangleq\mathbbm{1}\ppp{\s{T}_{\s{sum}}(\s{X})\geq\tau_{\s{sum}}},\label{eqn:sumTest}\\
    \calA_{\s{scan}}^{\s{SD}}(\s{X})&\triangleq\mathbbm{1}\ppp{\s{T}_{\s{scan}}^{\s{SD}}(\s{X})\geq \tau^{\s{SD}}_{\s{scan}}},\label{eqn:secantestTT}\\
    \calA_{\s{scan}}^{\s{CSD}}(\s{X})&\triangleq\mathbbm{1}\ppp{\s{T}_{\s{scan}}^{\s{CSD}}(\s{X})\geq \tau^{\s{CSD}}_{\s{scan}}},\label{eqn:secantestTT2}
\end{align}
where the thresholds are given by $\tau_{\s{sum}}\triangleq\frac{mk^2\lambda}{2}$, $\tau^{\s{SD}}_{\s{scan}}\triangleq \sqrt{(4+\delta)k^2\log\binom{n}{k}}$, and $\tau^{\s{CSD}}_{\s{scan}}\triangleq \sqrt{(4+\delta)k^2\log{n}}$, and correspond roughly to the average between the expected values of each of the statistics in \eqref{eqn:secantestSum}--\eqref{eqn:secantest0CSD} under the null and alternative hypotheses. It should be emphasized that the tests in \eqref{eqn:sumTest}--\eqref{eqn:secantestTT} were proposed in, e.g., \cite{kolar2011minimax,butucea2013detection,ma2015computational}, for the single planted submatrix detection problem.

A few important remarks are in order. First, note that in the scan test, we search for a single planted matrix rather than $m$ such matrices. Second, the sum test exhibits polynomial computational complexity, of $O(n^2)$ operations, and hence efficient. The scan test in \eqref{eqn:secantestTT}, however, exhibits an exponential computational complexity, and thus is inefficient. Indeed, the search space in \eqref{eqn:secantestTT} is of cardinality $|{\calK}_{k,1,n}|=\binom{n}{k}^2$. On the other hand, the scan test $\calA_{\s{scan}}^{\s{CSD}}$ for the consecutive setting is efficient because $|\calK_{k,1,n}^{\s{con}}|\leq n^2$. 
The following result provides sufficient conditions under which the risk of each of the above tests is asymptotically small. 

\begin{theorem}[Detection upper bounds]\label{thm:upper}
Consider the detection problems in Definitions~\ref{def:SD} and \ref{def:gcons}. Then, we have the following bounds:
\begin{enumerate}
    \item (Efficient $\s{SD}$) There exists an efficient algorithm $\calA_{\s{sum}}$ in \eqref{eqn:sumTest}, such that if
    \begin{align}
    \lambda = \omega\p{\frac{n}{mk^2}},
\end{align}
    then $\s{R}\p{\calA_{\s{sum}}}\to0$, as $n\to\infty$, for the problems in Definitions~\ref{def:SD} and \ref{def:gcons}.
\item (Exhaustive $\s{SD}$) There exists an algorithm $\calA_{\s{scan}}^{\s{SD}}$ in \eqref{eqn:secantestTT}, such that if
    \begin{align}
    \lambda = \omega\p{\sqrt{\frac{\log\frac{n}{k}}{k}}},
\end{align}
    then $\s{R}\p{\calA_{\s{scan}}^{\s{SD}}}\to0$, as $n\to\infty$, for the problem in Definition~\ref{def:SD}. 
\item (Efficient $\s{CSD}$) There exists an efficient algorithm $\calA_{\s{scan}}^{\s{CSD}}$ in \eqref{eqn:secantestTT2}, such that if
    \begin{align}
    \lambda = \omega\p{\frac{\sqrt{\log\frac{n}{k}}}{k}},
\end{align}
    then $\s{R}(\calA_{\s{scan}}^{\s{CSD}})\to0$, as $n\to\infty$, for the problem in Definition~\ref{def:gcons}.
\end{enumerate}
\end{theorem}

As can be seen from Theorem~\ref{thm:upper}, only the sum test performance barrier exhibits dependency on $m$. The scan test is, for both $\s{SD}$ and $\s{CSD}$, inherently independent of $m$. This makes sense because when summing all the elements of $\s{X}$, as $m$ gets larger the mean (the ``signal") under the alternative hypothesis gets larger as well. On the other hand, since the scan test searches for a single planted submatrix, the number of planted submatrices does not play a role. One could argue that it might be beneficial to search for the $m$ planted submatrices in the scan test, however, as we show below, this is not needed, and the bounds above are asymptotically tight. 

\paragraph{Lower bounds.} To present our lower bounds, we first recall that the optimal testing error probability is determined by the total variation distance between the distributions under the null and the alternative hypotheses as follows (see, e.g., \cite[Lemma 2.1]{Tsybakov}),
\begin{align}
    \min_{\calA_n:\mathbb{R}^{n\times n}\to\ppp{0,1}}\pr_{\calH_0}(\calA_n(\s{X})=1)+\pr_{\calH_1}(\calA_n(\s{X})=0) = 1-d_{\s{TV}}(\pr_{\calH_0},\pr_{\calH_1}).
\end{align}
The following result shows that under certain conditions the total variation between the null and alternative distributions is asymptotically small, and thus, there exists no test which can solve the above detection problems reliably.
\begin{theorem}[Information-theoretic lower bounds]\label{thm:lower}
We have the following results.
\begin{enumerate}
    \item Consider the detection problem in Definition~\ref{def:SD}. 
If,
\begin{align}
\lambda = o\p{\frac{n}{mk^2}\wedge \frac{1}{\sqrt{k}}},
\end{align}
then $d_{\s{TV}}(\pr_{\calH_0},\pr_{\calH_1})=o(1)$.
\item Consider the detection problem in Definition~\ref{def:gcons}. If $\lambda = o\p{k^{-1}}$,
then $d_{\s{TV}}(\pr_{\calH_0},\pr_{\calH_1})=o(1)$.
\end{enumerate}
 \end{theorem}

Theorem~\ref{thm:lower} above shows that our upper bounds in Theorem~\ref{thm:upper} are tight up to poly-log factors. Indeed, item 1 in Theorem~\ref{thm:lower} complements Items 1-2 in Theorem~\ref{thm:upper}, for the $\s{SD}$ problem, while item 2 in Theorem~\ref{thm:lower} complements Item 3 in Theorem~\ref{thm:upper}, for the $\s{CSD}$ problem. In the sequel, we illustrate our results using phase diagrams that show the tradeoff between $k$ and $\lambda$ as a function of $n$. One evident and important observation here is that the statistical limit for the $\s{CSD}$ problem is attained using an efficient test. Thus, there is no statistical computational gap in the detection problem in Definition~\ref{def:gcons}, and accordingly, it is either statistically impossible to solve the detection problem or it can be solved in polynomial time. This is not the case for the $\s{SD}$ problem. Note that both the efficient sum and the exhaustive scan tests are needed to attain the information-theoretic lower bound (up to poly-log factors). As discussed above, however, here the scan test is not efficient. We next give evidence that, based on the low-degree polynomial conjecture, efficient algorithms that run in polynomial-time do not exist in the regime where the scan test succeeds while the sum test fails. 

\paragraph{Computational lower bounds.} Note that the problem in Definition~\ref{def:SD} exhibits a gap in terms of what can be achieved by the proposed polynomial-time algorithm and the computationally expensive scan test algorithm. In particular, it can be seen that in the regime where $\frac{1}{\sqrt{k}}\ll\lambda\ll \frac{n}{mk^2}$, while the problem can be solved by an exhaustive search using the scan test, we do not have a polynomial-time algorithm. Next, we give evidence that, in fact, an efficient algorithm does not exist in this region. To that end, we start with a brief introduction to the method of \emph{low-degree polynomials}. 

The premise of this method is to take low-degree multivariate polynomials in the entries of the observations as a proxy for efficiently-computable functions. The ideas below were first developed in a sequence of works in the sum-of-squares optimization literature \cite{barak2016nearly,Hopkins18,hopkins2017bayesian,hopkins2017power}.

In the following, we follow the notations and definitions of~\cite{Hopkins18,Dmitriy19}. Any distribution $\pr_{\calH_0}$ on $\Omega_n=\mathbb{R}^{n\times n}$ induces an inner product of measurable functions $f,g:\Omega_n\to\mathbb{R}$ given by $\left\langle f,g \right\rangle_{\calH_0} = \bE_{\calH_0}[f(\s{X})g(\s{X})]$, and norm $\norm{f}_{\calH_0} = \left\langle f,f \right\rangle_{\calH_0}^{1/2}$. We Let $L^2(\pr_{\calH_0})$ denote the Hilbert space consisting of functions $f$ for which $\norm{f}_{\calH_0}<\infty$, endowed with the above inner product and norm. In the computationally-unbounded case, the Neyman-Pearson lemma shows that the likelihood ratio test achieves the optimal tradeoff between $\mathsf{Type}$-$\mathsf{I}$ and $\mathsf{Type}$-$\mathsf{II}$ error probabilities. Furthermore, it is well-known that the same test optimally distinguishes $\pr_{\calH_0}$ from $\pr_{\calH_1}$ in the $L^2$ sense. 
Specifically, denoting by $\s{L}_n\triangleq\pr_{\calH_1}/\pr_{\calH_0}$ the likelihood ratio, the second-moment method for contiguity (see, e.g., \cite{Dmitriy19}) shows that if $\norm{\s{L}_n}_{\calH_0}^2$ remains bounded as $n\to\infty$, then $\pr_{\calH_1}$ is contiguous to $\pr_{\calH_0}$. This implies that $\pr_{\calH_1}$ and $\pr_{\calH_0}$ are statistically indistinguishable, i.e., no test can have both $\mathsf{Type}$-$\mathsf{I}$ and $\mathsf{Type}$-$\mathsf{II}$ error probabilities tending to zero. 

We now describe the low-degree method. The idea is to find the low-degree polynomial that best distinguishes $\pr_{\calH_0}$ from $\pr_{\calH_1}$ in the $L^2$ sense. To that end, we let $\calV_{n,\leq\s{D}}\subset L^2(\pr_{\calH_0})$ denote the linear subspace of polynomials $\Omega_n\to\mathbb{R}$ of degree at most~$\s{D}\in\mathbb{N}$. We further define  $\calP_{\leq \s{D}}: L^2(\pr_{\calH_0})\to\calV_{n,\leq\s{D}}$ the orthogonal projection operator. Then, the \emph{$\s{D}$-low-degree likelihood ratio} $\s{L}_{n}^{\leq \s{D}}$ is the projection of a function $\s{L}_{n}$ to the span of coordinate-degree-$\s{D}$ functions, where the projection is orthogonal with respect to the inner product $\left\langle \cdot,\cdot \right\rangle_{\calH_0}$. As discussed above, the likelihood ratio optimally distinguishes $\pr_{\calH_0}$ from $\pr_{\calH_1}$ in the $L^2$ sense. The next lemma shows that over the set of low-degree polynomials, the~$\s{D}$-low-degree likelihood ratio have exhibit the same property.
\begin{lemma}[Optimally of $\s{L}_{n}^{\leq \s{D}}$ {\cite{hopkins2017bayesian,hopkins2017power,Dmitriy19}}]\label{lem:Dmitriy}
Consider the following optimization problem:
\begin{equation}
\begin{aligned}
\mathrm{max}
\;\bE_{\calH_1}f(\s{X})
\quad\mathrm{s.t.}
\quad\bE_{\calH_0}f^2(\s{X}) = 1,\; f\in\calV_{n,\leq\s{D}}.
\end{aligned}\label{eqn:optimizationProblem}
\end{equation}
Then, the unique solution $f^\star$ for \eqref{eqn:optimizationProblem} is the $\s{D}$-low degree likelihood ratio $f^\star = \s{L}_{n}^{\leq \s{D}}/\norm{\s{L}_{n}^{\leq \s{D}}}_{\calH_0}$, and the value of the optimization problem is $\norm{\s{L}_{n}^{\leq \s{D}}}_{\calH_0}$. 
\end{lemma}
\vspace{-0.2cm}
As was mentioned above, in the computationally-unbounded regime, an important property of the likelihood ratio is that if $\norm{\s{L}_n}_{\calH_0}$ is bounded, then $\pr_{\calH_0}$ and $\pr_{\calH_1}$ are statistically indistinguishable. The following conjecture states that a computational analog of this property holds, with $\s{L}_{n}^{\leq \s{D}}$ playing the role of the likelihood ratio. In fact,  it also postulates that polynomials of degree $\approx\log n$ are a proxy for polynomial-time algorithms. The conjecture below is based on \cite{Hopkins18,hopkins2017bayesian,hopkins2017power}, and~\cite[Conj. 2.2.4]{Hopkins18}. We give an informal statement of this conjecture, which appears in \cite[Conj. 1.16]{Dmitriy19}. For a precise statement, we refer the reader to~\cite[Conj. 2.2.4]{Hopkins18} and \cite[Sec. 4]{Dmitriy19}.

\begin{conjecture}[Low-degree conjecture, informal]\label{conj:1}
Given a sequence of probability measures $\pr_{\calH_0}$ and $\pr_{\calH_1}$, if there exists $\epsilon>0$ and $\s{D} = \s{D}(n)\geq (\log n)^{1+\epsilon}$, such that $\norm{\s{L}_{n}^{\leq \s{D}}}_{\calH_0}$ remains bounded as $n\to\infty$, then there is no polynomial-time algorithm that distinguishes~$\pr_{\calH_0}$ and $\pr_{\calH_1}$.
\end{conjecture}
\vspace{-0.2cm}

In the sequel, we will rely on Conjecture~\ref{conj:1} to give evidence for the statistical-computational gap observed for the problem in Definition~\ref{def:SD} in the regime where $\frac{1}{\sqrt{k}}\ll\lambda\ll \frac{n}{mk^2}$. At this point we would like to mention \cite[Hypothesis 2.1.5]{Hopkins18}, which states a more general form of Conjecture~\ref{conj:1} in the sense that it postulates that degree-$\s{D}$ polynomials are a proxy for $n^{O(D)}$-time algorithms. Note that if $\norm{\s{L}_{n}^{\leq \s{D}}}_{\calH_0} = O(1)$, then we expect detection in time $\s{T}(n) = e^{\s{D}(n)}$ to be impossible.

\begin{theorem}[Computational lower bound]\label{thm:gap}
Consider the detection problem in Definition~\ref{def:SD}. Then, if $\lambda$ is such that $\frac{1}{\sqrt{k}}\ll\lambda\ll \frac{n}{mk^2}$, then $\norm{\s{L}_{n}^{\leq \s{D}}}_{\calH_0}\leq O(1)$, for any $\s{D} = \Omega(\log n)$. On the other hand, if $\lambda$ is such that $\lambda\gg \frac{n}{mk^2}$, then $\norm{\s{L}_{n}^{\leq \s{D}}}_{\calH_0}\geq \omega(1)$.
\end{theorem}
Together with Conjecture~\ref{conj:1}, Theorem~\ref{thm:gap} implies that if we take degree-$\log n$ polynomials as a proxy for all efficient algorithms, our calculations predict that an $n^{O(\log n)}$ algorithm does not exist when $\frac{1}{\sqrt{k}}\ll\lambda\ll \frac{n}{mk^2}$. This is summarized in the following corollary.
\begin{corollary}
Consider the detection problem in Definition~\ref{def:SD}, and assume that Conjecture~\ref{conj:1} holds. An $n^{O(\log n)}$ algorithm that achieves strong detection does not exist if $\lambda$ is such that $\frac{1}{\sqrt{k}}\ll\lambda\ll \frac{n}{mk^2}$.
\end{corollary}
These predictions agree precisely with the previously established statistical-computational tradeoffs in the previous subsections. A more explicit formula for the computational barrier which exhibits dependency on $\s{D}$ and $\lambda$ can be deduced from the proof of Theorem~\ref{thm:gap}; to keep the exposition simple we opted to present the refined result above. 

We note that numerical and theoretical evidence for the existence of computational-statistical gaps were observed in other statistical models that are also inspired by cryo-EM, including heterogeneous multi-reference alignment~\cite{boumal2018heterogeneous,wein2018statistical} and sparse multi-reference alignment~\cite{bendory2022sparse}.

\paragraph{Phase diagrams.} Using Theorems~\ref{thm:upper}--\ref{thm:gap} we are now in a position to draw the obtained phase diagrams for our detection problems. Specifically, treating $k$ and $\lambda$ as polynomials in~$n$, i.e., $k = \Theta(n^\beta)$ and $\lambda = \Theta(n^{-\alpha})$, for some $\alpha\in(0,1)$ and $\beta\in(0,1)$, we obtain the phase diagrams in Figure~\ref{fig:spcaphasediagram1}, for a fixed number of submatrices $m=O(1)$. Specifically,
\begin{enumerate}
    \item \emph{Computationally easy regime (blue region):} there is a polynomial-time algorithm for the detection task when $\alpha<2\beta-1$.
    \item \emph{Computationally hard regime (red region):} there is an inefficient algorithm for detection when $\alpha<\beta/2$ and $\alpha>2\beta-1$, but the problem is computationally hard (no polynomial-time algorithm exists) in the sense that the class of low-degree polynomials fails in this region.
    \item \emph{Statistically impossible regime:} detection is statistically impossible when $\alpha>\frac{\beta}{2}\vee (2\beta-1)$.
\end{enumerate}
\begin{figure}
     \hspace{1.2cm}\begin{subfigure}[b]{0.3\textwidth}
     \centering
    \begin{tikzpicture}[scale=1.5]
\tikzstyle{every node}=[font=\footnotesize]
\def\xmin{0}
\def\xmax{3.5}
\def\ymin{0}
\def\ymax{3.5}

\draw[->] (\xmin,\ymin) -- (\xmax,\ymin) node[right] {$\beta$};
\draw[->] (\xmin,\ymin) -- (\xmin,\ymax) node[above] {$\alpha$};

\node at (3, 0) [below] {$1$};
\node at (1.5, 0) [below] {$\frac{1}{2}$};
\node at (0, 3) [left] {$1$};
\node at (0, 1) [left] {$\frac{1}{3}$};
\node at (2, 0) [below] {$\frac{2}{3}$};
\node at (0, 0) [left] {$0$};

\filldraw[fill=gray!25, draw=black] (0, 0) -- (0, 3) -- (3, 3) -- (2, 1) -- (0,0);
\filldraw[fill=red!25, draw=black] (0, 0) -- (2, 1) -- (1.5, 0) -- (0, 0);
\filldraw[fill=blue!25, draw=black] (1.5, 0) -- (3, 3) -- (3, 0) -- (1.5, 0);
\node at (0.9, 1.6) {``$\s{\mathbf{Statistically}}$};
\node at (1.1, 1.3) {$\s{\mathbf{Impossible}}$"};
\node at (2.5, 0.8) {``$\s{\mathbf{Easy}}$"};
\node at (0.99, 0.2) {``$\s{\mathbf{Hard}}$"};
\draw [dashed] (0,1) -- (2,1);
\draw [dashed] (2,0) -- (2,1);
\end{tikzpicture}
     \caption{$m=O(1)$}
\label{fig:spcaphasediagram1}
     \end{subfigure}\quad\quad\quad\quad\quad\quad\quad\quad\begin{subfigure}[b]{0.3\textwidth}
     \centering
     \begin{tikzpicture}[scale=1.5]
\tikzstyle{every node}=[font=\footnotesize]
\def\xmin{0}
\def\xmax{3.5}
\def\ymin{0}
\def\ymax{4.25}

\draw[->] (\xmin,\ymin) -- (\xmax,\ymin) node[right] {$\beta$};
\draw[->] (\xmin,\ymin) -- (\xmin,\ymax) node[above] {$\alpha$};

\node at (3, 0) [below] {$1$};
\node at (9/8, 0) [below] {$\frac{3}{8}$};
\node at (0, 3.75) [left] {$\frac{5}{4}$};
\node at (0, 0.75) [left] {$\frac{1}{4}$};
\node at (1.5, 0) [below] {$\frac{1}{2}$};
\node at (0, 0) [left] {$0$};

\filldraw[fill=gray!25, draw=black] (0, 0) -- (0, 3.75) -- (3, 3.75) -- (1.5, 0.75) -- (0,0);
\filldraw[fill=red!25, draw=black] (0, 0) -- (1.5, 0.75) -- (9/8, 0) -- (0, 0);
\filldraw[fill=blue!25, draw=black] (9/8, 0) -- (3, 3.75) -- (3, 0) -- (9/8, 0);
\node at (0.9, 2.2) {``$\s{\mathbf{Statistically}}$};
\node at (1.1, 1.9) {$\s{\mathbf{Impossible}}$"};
\node at (2.3, 0.8) {``$\s{\mathbf{Easy}}$"};
\node[rotate=30] at (0.95, 0.3) {``$\s{\mathbf{Hard}}$"};
\draw [dashed] (0,0.75) -- (1.5,0.75);
\draw [dashed] (1.5,0) -- (1.5,0.75);
\end{tikzpicture}
\caption{$m=\Theta(n^{1/4})$}
\label{fig:spcaphasediagram2}
     \end{subfigure}
     \caption{Phase diagrams for submatrix detection as a function of $k=\Theta(n^{\beta})$, and $\lambda= \Theta(n^{-\alpha})$, for $m=O(1)$ and $m=\Theta(n^{1/4})$.}
        \label{fig:subDetection}
\end{figure}
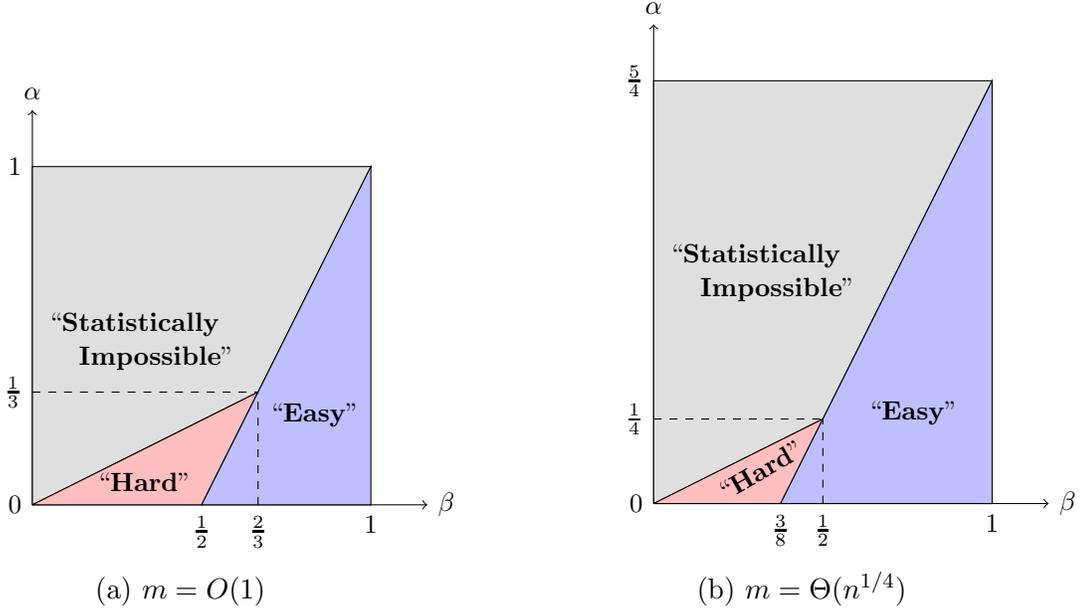

When the number of submatrices grows with $n=\omega(1)$, we get different phase diagrams depending on its value. For example, if $m=\Theta(n^{1/4})$, we get Figure~\ref{fig:spcaphasediagram2}. Specifically,
\begin{enumerate}
    \item \emph{Computationally easy regime (blue region):} there is a polynomial-time algorithm for the detection task when $\alpha<2\beta-\frac{3}{4}$.
    \item \emph{Computationally hard regime (red region):} there is an inefficient algorithm for detection when $\alpha<\beta/2$ and $\alpha>2\beta-\frac{3}{4}$, but the problem is computationally hard (no polynomial-time algorithm exists) in the sense that the class of low-degree polynomials fails in this region.
    \item \emph{Statistically impossible regime:} detection is statistically impossible when $\alpha>\frac{\beta}{2}\vee (2\beta-3/4)$.
\end{enumerate}

Finally, for the consecutive problem, we get the phase diagram in Figure~\ref{fig:spcaphasediagram3}, independently of the value of $m$. Here, there are only two regions where the problem is either statistically impossible or easy to solve. 

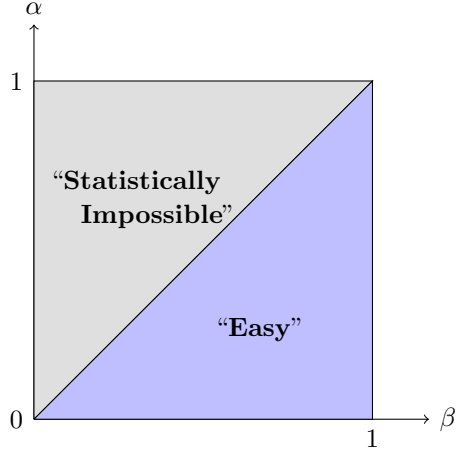
\begin{figure}[t]
\centering

\begin{tikzpicture}[scale=1.5]
\tikzstyle{every node}=[font=\footnotesize]
\def\xmin{0}
\def\xmax{3.5}
\def\ymin{0}
\def\ymax{3.5}

\draw[->] (\xmin,\ymin) -- (\xmax,\ymin) node[right] {$\beta$};
\draw[->] (\xmin,\ymin) -- (\xmin,\ymax) node[above] {$\alpha$};

\node at (3, 0) [below] {$1$};
\node at (0, 3) [left] {$1$};
\node at (0, 0) [left] {$0$};

\filldraw[fill=gray!25, draw=black] (0, 0) -- (0, 3) -- (3, 3) --(0,0);
\filldraw[fill=blue!25, draw=black] (0, 0) -- (3, 0) -- (3, 3) -- (0, 0);
\node at (0.9, 2.1) {``$\s{\mathbf{Statistically}}$};
\node at (1.1, 1.8) {$\s{\mathbf{Impossible}}$"};
\node at (2, 0.8) {``$\s{\mathbf{Easy}}$"};
\end{tikzpicture}

\caption{Phase diagram for consecutive submatrix detection, as a function of $k=\Theta(n^{\beta})$, and $\lambda= \Theta(n^{-\alpha})$, for any $m$.}
\label{fig:spcaphasediagram3}
\end{figure}%

\subsection{The recovery problem}

\paragraph{Upper bounds.} We start by presenting our upper bounds for both exact and correlated types of recovery for the consecutive problem in Definition~\ref{def:gconsrec}. To that end, we propose the following recovery algorithm. It can be shown that the maximum-likelihood (ML) estimator, minimizing the error probability, is given by (see Subsection~\ref{subsec:MLEDerivation} for a complete derivation),
\begin{align}
    \hat{\s{K}}_{\s{ML}}(\s{X}) = \argmax_{\s{K}\in\calK_{k,m,n}^{\s{con}}}\sum_{(i,j)\in\s{K}}\s{X}_{ij}.\label{eqn:MLest0}
\end{align}
The computational complexity of the exhaustive search in \eqref{eqn:MLest0} is of order $n^{2m}$.
Thus, for $m = O(1)$, the ML estimator runs in polynomial time, and thus, is efficient. However, if $m=\omega(1)$ then the exhaustive search is not efficient anymore. Nonetheless, the following straightforward modification of \eqref{eqn:MLest0} provably achieves the same asymptotic performance of the ML estimator above, and at the same time computationally efficient. 

Before we present this algorithm, we make a simplifying technical assumption on the possible set of planted submatrices, and then explain how this assumption can be removed. \emph{We assume that each pair of submatrices in the underlying planted submatrices $\s{K}^\star$ are at least $k$ columns and rows far way}. In other words, there are at least $k$ columns and $k$ rows separating any pair of submatrices in $\s{K}^\star$. Similar assumptions are frequently taken when analyzing statistical models inspired by cryo-EM, see, for example~\cite{bendory2018toward}. We will refer to the above as the \emph{separation assumption}.

Our recovery algorithm works as follows: in the $\ell\in[m]$ step, we find the ML estimate of a single submatrix using,
\begin{align}
    \hat{\s{K}}_{\ell}(\s{X}^{(\ell)}) = \argmax_{\s{K}\in\calK_{k,1,n}^{\s{con}}}\sum_{(i,j)\in\s{K}}\s{X}^{(\ell)}_{ij},\label{eqn:MLPeel}
\end{align}
where $\s{X}^{(\ell)}$ is defined recursively as follows: $\s{X}^{(1)}\triangleq\s{X}$, and for $\ell\geq2$,
\begin{align}
    \s{X}^{(\ell)} = \s{X}^{(\ell-1)}\odot\s{E}(\hat{\s{K}}_{\ell-1}),
\end{align}
where $\s{E}(\hat{\s{K}}_{\ell-1})$ is an $n\times n$ matrix such that $[\s{E}(\hat{\s{K}}_{\ell-1})]_{ij} = -\infty$, for $(i,j)\in\hat{\s{K}}_{\ell-1}$, and $[\s{E}(\hat{\s{K}}_{\ell-1})]_{ij} = 1$, otherwise. To wit, in each step of the algorithm we ``peel" the set of estimated indices (or, estimated submatrices) in previous steps from the search space. This is done by setting the corresponding entries of $\s{X}$ to $-\infty$ so that the sum in \eqref{eqn:MLPeel} will not be maximized by previously chosen sets of indices. We denote by $\hat{\s{K}}_{\s{peel}}(\s{X}) = \{\hat{\s{K}}_{\ell}\}_{\ell=1}^{m}$ the output of the above algorithm.
    \begin{rmk}
Without the assumption above, the fact that the peeling algorithm succeeds is not trivial. If, for example, the chosen planted matrices are such that they include a pair of adjacent matrices, then it could be the case that at some step of the peeling algorithm, the estimated set of indices corresponds to a certain submatrix of the union of those adjacent matrices. However, one can easily modify the peeling algorithm, drop the assumption above, and obtain the same statistical guarantees stated below. Indeed, consider the following modification to the peeling routine in Algorithm~\ref{algo:peel}.

\begin{figure}[ht!]
  \centering
  \begin{minipage}{.7\linewidth}
\begin{algorithm}[H]
  \caption{\texttt{Modified Peeling}}\label{algo:peel}
  \begin{enumerate}[topsep=0pt,itemsep=-1ex,partopsep=1ex,parsep=1ex]
    \item 
    \textbf{Initialize} $\s{flag}\leftarrow0$, $\ell\leftarrow1$, $\mathscr{K}\leftarrow\emptyset$, $\s{A} = \mathbf{0}_{n\times n}$.
    \item
    \textbf{while} $\s{flag}=0$
    \begin{enumerate}[topsep=0pt,itemsep=-1ex,partopsep=1ex,parsep=1ex]
     \item $\hat{\s{K}}_{\ell}(\s{X}) \leftarrow \argmax_{\s{K}\in\calK_{k,1,n}^{\s{con}}\setminus\mathscr{K}}\sum_{(i,j)\in\s{K}}\s{X}_{ij}.$
     \item $\s{A}_{ij}\leftarrow1$, for $(i,j)\in\hat{\s{K}}_{\ell}(\s{X})$, and $\s{A}_{ij}\leftarrow0$, otherwise.
     \item $\mathscr{K}\leftarrow\mathscr{K}\cup\hat{\s{K}}_{\ell}(\s{X}).$
     \item \textbf{if} $\innerP{\Jb,\s{A}}=mk^2$
     
        \vspace{-1ex}
        $\s{flag}\leftarrow1.$
     \item \textbf{else} 

     \vspace{-1ex}
     $\ell\leftarrow \ell+1.$
    \end{enumerate}

  \item \textbf{Output} $\s{A}$. 
  \end{enumerate}
\end{algorithm}
\end{minipage}
\end{figure}

The key idea is as follows. In the first step, we find the $k\times k$ submatrix in $\s{X}$ with the maximum sum of entries. We denote this submatrix by $\hat{\s{K}}_1$. This is exactly the same first step of the peeling algorithm. In the second step, we again search for the $k\times k$ submatrix in~$\s{X}$ with the maximum sum of entries, but of course, remove $\hat{\s{K}}_1$ from the search space. More generally, in the $\ell$-th step, we again search for the $k\times k$ submatrix in $\s{X}$ with maximum sum of entries, but remove $\mathscr{K}=\cup_{i=1}^{\ell-1}\hat{\s{K}}_i$ from the search space. We terminate this process once $\cup_{i=1}^{\ell}\hat{\s{K}}_i\in\calK_{k,m,n}^{\s{con}}$, i.e., the union of the estimated sets of matrices can cast as a proper set of planted submatrices. This can easily be checked by forming the matrix $\s{A}$ in Step 2(b), and checking the conditions in Step 2(d). If the actually planted submatrices are not adjacent, then this will be the case (under the conditions in the theorem below) after $\ell=m$ steps, with high probability. Otherwise, if at least two planted submatrices are adjacent, then while $\ell$ might be larger than $m$ it is bounded by $n^2$, and it is guaranteed that such a union exists. Once we find such a union, it is easy to revert the set of $m$ consecutive $k\times k$ submatrices from $\s{A}$.    
\end{rmk}
We have the following result.
\begin{theorem}[Recovery upper bounds]\label{thm:recupper}
Consider the recovery problem in Definition~\ref{def:gconsrec}, and let $\s{C}$ be a universal constant. Then, we have the following set of bounds:
    \begin{enumerate}
        \item (ML Exact Recovery) Consider the ML estimator in \eqref{eqn:MLest0}. If
        \begin{align}
    \liminf_{n\to\infty}\frac{\lambda}{\sqrt{\s{C}k^{-1}\log n}}>1,
\end{align}
        then exact recovery is possible.
        \item (Peeling Exact Recovery) Consider the peeling estimator in \eqref{eqn:MLPeel}, and assume that the separation assumption holds. Then, if 
        \begin{align}
    \liminf_{n\to\infty}\frac{\lambda}{\sqrt{\s{C}k^{-1}\log n}}>1,
\end{align}
        then exact recovery is possible.

    \item (Peeling Correlated Recovery) Consider the peeling estimator in \eqref{eqn:MLPeel}, and assume that the separation assumption holds. If
    \begin{align}
    \liminf_{n\to\infty}\frac{\lambda}{\sqrt{\s{C}k^{-2}\log n}}>1,
\end{align}
    then correlated recovery is possible.
    \end{enumerate}
\end{theorem}

\paragraph{Lower bounds.} The following result shows that under certain conditions, exact and correlated recoveries are impossible.
\begin{theorem}[Information-theoretic recovery lower bounds]\label{thm:recLower} Consider the recovery problem in Definition~\ref{def:gconsrec}. Then:
\begin{enumerate}
    \item If $\lambda<\s{C}\sqrt{\frac{\log m}{k}}$, exact recovery is impossible, i.e., 
    $$
\inf_{\hat{\s{K}}}\sup_{\s{K}^\star\in\calK_{k,m,n}^{\s{con}}}\pr[\hat{\s{K}}(\s{X})\neq\s{K}^\star]>\frac{1}{2},
    $$
    where the infimum ranges over all measurable functions of the matrix $\s{X}$.
    \item If $\lambda = o(k^{-1})$, correlated recovery is impossible, i.e., $\sup_{\s{K}^\star\in\calK_{k,m,n}^{\s{con}}}\s{overlap}(\s{K}^\star,\hat{\s{K}}) = o(mk^2)$.
\end{enumerate}
\end{theorem}
Thus, similarly to the detection problem, the consecutive recovery problem is either statistically impossible or easy to solve. The corresponding phase diagram for exact and correlated types of recoveries is given in Figure~\ref{fig:spcaphasediagram4}. Roughly speaking, exact recovery is possible if $\lambda =\omega(k^{-1/2})$ and impossible if $\lambda =o(k^{-1/2})$. Correlated recovery is possible if $\lambda =\omega(k^{-1})$ and impossible if $\lambda =o(k^{-1})$. 

A few remarks are in order. First, note that there is a gap between detection and exact recovery; the barrier for $\lambda$ for the former is at $k^{-1}$, while for the latter at $k^{-1/2}$. In the context of cryo-EM, this indicates a gap between the ability to detect the existence of particle images in the data set, and the ability to perform successful particle picking (exact recovery). Recently, new computational methods were devised to elucidate molecular structures without  particle picking, thus bypassing the limit of exact recovery,  allowing constructing structures in very low SNR environments, e.g., ~\cite{bendory2018toward,kreymer2022two,kreymer2023stochastic}. This in turn opens the door to  recovering small molecular structures that induce low SNR~\cite{henderson1995potential}. Second, there is no gap between detection and correlated recovery, and these different tasks are asymptotically statistically the same. The same gap exists between correlated and exact recoveries, implying that exact recovery is strictly harder than correlated recovery. 

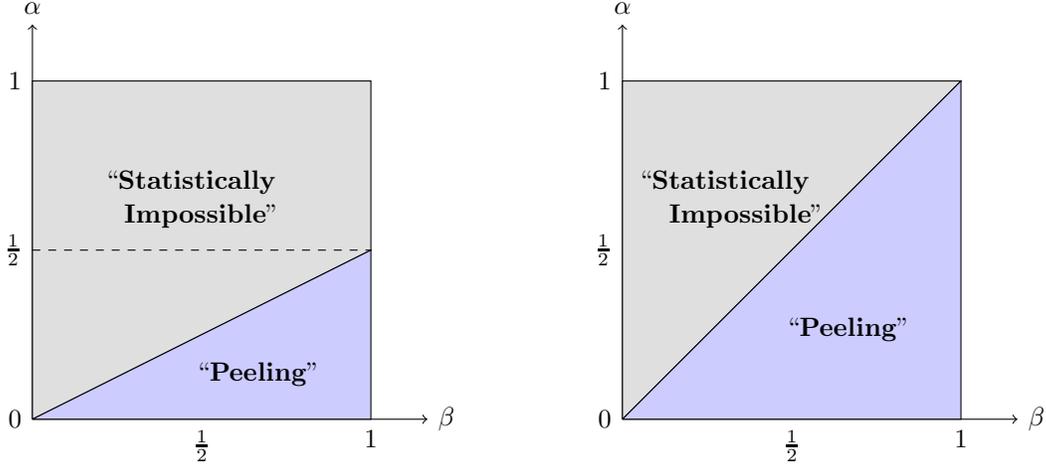
\begin{figure}
     \hspace{1cm}\begin{subfigure}[b]{0.3\textwidth}
     \centering

     \begin{tikzpicture}[scale=1.5]
\tikzstyle{every node}=[font=\footnotesize]
\def\xmin{0}
\def\xmax{3.5}
\def\ymin{0}
\def\ymax{3.5}

\draw[->] (\xmin,\ymin) -- (\xmax,\ymin) node[right] {$\beta$};
\draw[->] (\xmin,\ymin) -- (\xmin,\ymax) node[above] {$\alpha$};

\node at (3, 0) [below] {$1$};
\node at (0, 3) [left] {$1$};
\node at (0, 0) [left] {$0$};
\node at (1.5, 0) [below] {$\frac{1}{2}$};
\node at (0, 1.5) [left] {$\frac{1}{2}$};

\filldraw[fill=gray!25, draw=black] (0, 0) -- (0, 3) -- (3, 3) --(3,1.5)--(0,0);
\filldraw[fill=blue!20, draw=black] (0, 0) -- (3, 1.5) -- (3, 0) -- (0, 0);
\node at (1.4, 2.1) {``$\s{\mathbf{Statistically}}$};
\node at (1.5, 1.8) {$\s{\mathbf{Impossible}}$"};
\node at (2, 0.4) {``$\s{\mathbf{Peeling}}$"};
\draw [dashed] (0,1.5) -- (3,1.5);
\end{tikzpicture}
     \end{subfigure}\quad\quad\quad\quad\quad\quad\quad\begin{subfigure}[b]{0.3\textwidth}
     \centering
\begin{tikzpicture}[scale=1.5]
\tikzstyle{every node}=[font=\footnotesize]
\def\xmin{0}
\def\xmax{3.5}
\def\ymin{0}
\def\ymax{3.5}

\draw[->] (\xmin,\ymin) -- (\xmax,\ymin) node[right] {$\beta$};
\draw[->] (\xmin,\ymin) -- (\xmin,\ymax) node[above] {$\alpha$};

\node at (3, 0) [below] {$1$};
\node at (0, 3) [left] {$1$};
\node at (0, 0) [left] {$0$};
\node at (1.5, 0) [below] {$\frac{1}{2}$};
\node at (0, 1.5) [left] {$\frac{1}{2}$};

\filldraw[fill=gray!25, draw=black] (0, 0) -- (0, 3) -- (3, 3) --(0,0);
\filldraw[fill=blue!20, draw=black] (0, 0) -- (3, 0) -- (3, 3) -- (0, 0);
\node at (0.9, 2.1) {``$\s{\mathbf{Statistically}}$};
\node at (1.1, 1.8) {$\s{\mathbf{Impossible}}$"};
\node at (2, 0.8) {``$\s{\mathbf{Peeling}}$"};
\end{tikzpicture}

     \end{subfigure}
        \caption{Phase diagram for consecutive submatrix exact recovery (left) and correlated recovery (right), as a function of $k=\Theta(n^{\beta})$, and $\lambda= \Theta(n^{-\alpha})$, for any $m$.}
        \label{fig:spcaphasediagram4}
\end{figure}


\section{Proofs} \label{sec:proofs}
\subsection{Proof of Theorem~\ref{thm:upper}}\label{sec:algo}

\subsubsection{Sum test}\label{subsec:sum}

Recall the sum test in \eqref{eqn:sumTest}, and let $\tau \triangleq \frac{mk^2\lambda}{2}$. Let us analyze the corresponding error probability. On the one hand, under $\calH_0$, it is clear that $\s{T}_{\s{sum}}(\s{X})\sim\calN(0,n^2)$. Thus,
\begin{align}
    \pr_{\calH_0}\p{\calA_{\s{sum}}(\s{X})=1} &= \pr_{\calH_0}\p{\s{T}_{\s{sum}}(\s{X})\geq\tau}\nonumber\\
    & = \pr(\calN(0,n^2)\geq\tau)\\
    &\leq \frac{1}{2}\exp\p{-\frac{\tau^2}{2n^2}}.\nonumber
\end{align}
On the other hand, under $\calH_1$, $\s{T}_{\s{sum}}(\s{X})\sim\calN(mk^2\lambda,n^2)$. Thus,
\begin{align}
    \pr_{\calH_1}\p{\calA_{\s{sum}}(\s{X})=0} \nonumber &= \pr_{\calH_1}\p{\s{T}_{\s{sum}}(\s{X})\leq\tau}\\
    & = \pr(\calN(mk^2\lambda,n^2)\leq\tau)\\
    &\leq \frac{1}{2}\exp\p{-\frac{(\tau-mk^2\lambda)^2}{2n^2}}.\nonumber
\end{align}
Substituting $\tau = \frac{mk^2\lambda}{2}$, we obtain that
\begin{align}
    \s{R}\p{\calA_{\s{sum}}}\leq \exp\p{-\frac{m^2k^4\lambda^2}{8n^2}}.
\end{align}
Thus, if $\frac{mk^2\lambda}{n}\to\infty$, then $\s{R}\p{\calA_{\s{sum}}}\to0$, as $n\to\infty$. Note that the analysis above holds true for both detection problems in Definitions~\ref{def:SD} and \ref{def:gcons}.

\subsubsection{Scan test}\label{subsec:scan}

Recall the scan test $\calA_{\s{scan}}^{\s{SD}}(\s{X})$ in \eqref{eqn:secantestTT}, and its consecutive version $\calA_{\s{scan}}^{\s{CSD}}(\s{X})$. Let us start by analyzing the error probability associated with $\calA^{\s{SD}}_{\s{scan}}$. For simplicity of notation, we let $\tau\triangleq\sqrt{(4+\delta)k^2\log\binom{n}{k}}$. On the one hand, under $\calH_0$, we have
\begin{align}
    \pr_{\calH_0}\p{\calA^{\s{SD}}_{\s{scan}}(\s{X})=1} &= \pr_{\calH_0}\p{\s{T}^{\s{SD}}_{\s{scan}}(\s{X})\geq\tau}\nonumber\\
    & \leq \binom{n}{k}^2\pr(\calN(0,k^2)\geq\tau)\\
    &\leq \frac{1}{2}\exp\p{2\log\binom{n}{k}-\frac{\tau^2}{2k^2}}. \nonumber
\end{align}
On the other hand, under $\calH_1$, we have
\begin{align}
    \pr_{\calH_1}\p{\calA^{\s{SD}}_{\s{scan}}(\s{X})=0} &= \pr_{\calH_1}\p{\s{T}^{\s{SD}}_{\s{scan}}(\s{X})\leq\tau}\\
    & \leq \pr(\calN(k^2\lambda,k^2)\leq\tau)\\
    &\leq \frac{1}{2}\exp\p{-\frac{(k^2\lambda-\tau)_+^2}{2k^2}}.
\end{align}
Thus, 
\begin{align}
    \s{R}\p{\calA^{\s{SD}}_{\s{scan}}}\leq \frac{1}{2}\exp\p{2\log\binom{n}{k}-\frac{\tau^2}{2k^2}}+\exp\p{-\frac{(k^2\lambda-\tau)_+^2}{2k^2}}.
\end{align}
Substituting $\tau = \sqrt{(4+\delta)k^2\log\binom{n}{k}}$, we get
\begin{align}
    \s{R}\p{\calA^{\s{SD}}_{\s{scan}}}\leq \frac{1}{2}\exp\p{-\frac{\delta}{2}\cdot \log\binom{n}{k}}+\exp\p{-k^2\frac{\p{\lambda-\frac{\tau}{k^2}}_+^2}{2}},
\end{align}
and thus $\s{R}\p{\calA^{\s{SD}}_{\s{scan}}}\to0$, as $n\to\infty$, provided that $\liminf_{n\to\infty}\frac{\lambda}{\sqrt{4k^{-1}\log\frac{n}{k}}}>1$, as claimed. Next, we analyze $\calA^{\s{CSD}}_{\s{scan}}$. Let $\tau_{\s{c}} \triangleq \sqrt{(4+\delta)k^2\log{n}}$. As above, we have
\begin{align}
    \pr_{\calH_0}\p{\calA^{\s{CSD}}_{\s{scan}}(\s{X})=1} &= \pr_{\calH_0}\p{\s{T}^{\s{CSD}}_{\s{scan}}(\s{X})\geq\tau_{\s{c}}}\\
    & \leq n^{2}\pr(\calN(0,k^2)\geq\tau_{\s{c}})\\
    &\leq \frac{1}{2}\exp\p{2\log{n}-\frac{\tau_{\s{c}}^2}{2k^2}}.
\end{align}
On the other hand, under $\calH_1$, the result remained intact:
\begin{align}
    \pr_{\calH_1}\p{\calA^{\s{CSD}}_{\s{scan}}(\s{X})=0} &= \pr_{\calH_1}\p{\s{T}^{\s{CSD}}_{\s{scan}}(\s{X})\leq\tau_{\s{c}}}\\
    & \leq \pr(\calN(k^2\lambda,k^2)\leq\tau_{\s{c}})\\
    &\leq \frac{1}{2}\exp\p{-\frac{(k^2\lambda-\tau_{\s{c}})_+^2}{2k^2}}.
\end{align}
Thus, 
\begin{align}
    \s{R}\p{\calA^{\s{CSD}}_{\s{scan}}}\leq \frac{1}{2}\exp\p{2\log{n}-\frac{\tau_{\s{c}}^2}{2k^2}}+\exp\p{-\frac{(k^2\lambda-\tau_{\s{c}})_+^2}{2k^2}}.
\end{align}
Substituting $\tau_{\s{c}} = \sqrt{(4+\delta)k^2\log{n}}$, for  $\delta>0$, we get
\begin{align}
    \s{R}\p{\calA^{\s{CSD}}_{\s{scan}}}\leq \frac{1}{2}\exp\p{-\frac{\delta}{2}\cdot \log{n}}+\exp\p{-k^2\frac{\p{\lambda-\frac{\tau_{\s{c}}}{k^2}}_+^2}{2}},
\end{align}
and thus $\s{R}\p{\calA^{\s{CSD}}_{\s{scan}}}\to0$, as $n\to\infty$, provided that $\liminf_{n\to\infty}\frac{\lambda}{4\sqrt{k^{-2}\log\frac{n}{k}}}>1$, as claimed.

\subsection{Proof of Theorem~\ref{thm:lower}}\label{sec:IT}

\subsubsection{Submatrix detection}

Recall that the optimal test $\calA_n^\ast$ that minimizes the risk is the likelihood ratio test defined as follows,
\begin{align}
\calA_n^\ast\left(\s{X}\right) \triangleq \Ind\ppp{\mathsf{L}_n\left(\s{X}\right) \geq 1},
\end{align}
where $\mathsf{L}_n\left(\s{X}\right) \triangleq\frac{\pr_{\calH_1}\left(\s{X}\right)}{\pr_{\calH_0}\left(\s{X}\right)}$. The optimal risk, denoted by $\s{R}^\ast = \s{R}(\calA_n^\ast)$, can be lower bounded using the Cauchy–Schwartz inequality as follows,
\begin{align}
\mathsf{R}^\ast &=1-\frac{1}{2}\mathbb{E}_{\calH_0}\abs{\mathsf{L}_n\left(\s{X}\right)-1}   \\
&\geq 1-\frac{1} {2}\sqrt{\mathbb{E}_{\calH_0}\left[\left(\mathsf{L}_n\left(\s{X}\right)-1\right)^2\right]}   \nonumber \\
&=1-\frac{1} {2}\sqrt{\mathbb{E}_{\calH_0}\left[\left(\mathsf{L}_n\left(\s{X}\right)\right)^2\right]-1}. \nonumber
\end{align}
Thus, in order to lower bound the risk, we need to upper bound $\mathbb{E}_{\calH_0}\left[\left(\mathsf{L}_n\left(\s{X}\right)\right)^2\right]$. Below, we provide a lower bound that holds for any pair of distributions $\calP$ and $\calQ$.
\begin{corollary}\label{cor:1}
The following holds:
\begin{align}\mathbb{E}_{\calH_0}\left[\left(\mathsf{L}_n\left(\s{X}\right)\right)^2\right] =  \mathbb{E}_{\s{K}\indep\s{K}'}\pp{(1+\chi^2(\calP||\calQ))^{\abs{\s{K} \cap \s{K}'}}}\leq \mathbb{E}_{\s{K}\indep\s{K}'}\left[e^{\chi^2(\calP||\calQ)\cdot \abs{\s{K} \cap \s{K}'}}\right],
\end{align}
where $\s{K}$ and $\s{K}'$ are two independent copies drawn uniformly at random from $\calK_{k,m,n}$ (or, $\bar{\calK}_{k,m,n}$), and
\begin{align}
   \chi^2(\calP||\calQ)\triangleq \bE_{X\sim\calQ}\pp{\frac{\calP(X)}{\calQ(X)}}^2-1.
\end{align}
\end{corollary}
\begin{proof}[Proof of Corollary~\ref{cor:1}]
First, note that the likelihood can be written as follows:
\begin{align}
    \mathsf{L}_n\left(\s{X}\right)  &= \frac{\pr_{\calH_1}(\s{X})}{\pr_{\calH_0}(\s{X})}= \bE_{\s{K}\sim\s{Unif}(\calK_{k,m,n})}\p{\prod_{(i,j)\in\s{K}}\frac{\calP(\s{X}_{ij})}{\calQ(\s{X}_{ij})}}.\label{eqn:averageoverK}
\end{align}
Now, note that the square of the right-hand side of \eqref{eqn:averageoverK} can be rewritten as:
\begin{align}
    \pp{\bE_{\s{K}\sim\s{Unif}(\calK_{k,m,n})}\p{\prod_{(i,j)\in\s{K}}\frac{\calP(\s{X}_{ij})}{\calQ(\s{X}_{ij})}}}^2 = \mathbb{E}_{\s{K}\indep\s{K}'\sim\s{Unif}(\calK_{k,m,n})}\p{\prod_{(i,j)\in\s{K}}\frac{\calP(\s{X}_{ij})}{\calQ(\s{X}_{ij})}\prod_{(i,j)\in\s{K}'}\frac{\calP(\s{X}_{ij})}{\calQ(\s{X}_{ij})}}.
\end{align}
Therefore,
\begin{align}
&\mathbb{E}_{\calH_0}\left[\left(\mathsf{L}_n\left(\s{X}\right)\right)^2\right]  = \mathbb{E}_{\calH_0}\pp{\bE_{\s{K}\sim\s{Unif}(\calK_{k,m,n})}\p{\prod_{(i,j)\in\s{K}}\frac{\calP(\s{X}_{ij})}{\calQ(\s{X}_{ij})}}}^2\\
&\hspace{1cm} = \mathbb{E}_{\calH_0}\pp{\mathbb{E}_{\s{K}\indep\s{K}'\sim\s{Unif}(\calK_{k,m,n})}\p{\prod_{(i,j)\in\s{K}}\frac{\calP(\s{X}_{ij})}{\calQ(\s{X}_{ij})}\prod_{(i,j)\in\s{K}'}\frac{\calP(\s{X}_{ij})}{\calQ(\s{X}_{ij})}}}\\
& \hspace{1cm}= \mathbb{E}_{\s{K}\indep\s{K}'\sim\s{Unif}(\calK_{k,m,n})}\pp{\mathbb{E}_{\calH_0}\p{\prod_{(i,j)\in\s{K}}\frac{\calP(\s{X}_{ij})}{\calQ(\s{X}_{ij})}\prod_{(i,j)\in\s{K}'}\frac{\calP(\s{X}_{ij})}{\calQ(\s{X}_{ij})}}}\\
& \hspace{1cm}= \mathbb{E}_{\s{K}\indep\s{K}'\sim\s{Unif}(\calK_{k,m,n})}\pp{\mathbb{E}_{\calH_0}\p{\prod_{(i,j)\in\s{K \cup K' \setminus K \cap K' }}\frac{\calP(\s{X}_{ij})}{\calQ(\s{X}_{ij})}\prod_{(i,j)\in\s{K\cap K'}}\p{\frac{\calP(\s{X}_{ij})}{\calQ(\s{X}_{ij})}}^2}} \\
& \hspace{1cm}= \mathbb{E}_{\s{K}\indep\s{K}'\sim\s{Unif}(\calK_{k,m,n})}\pp{\prod_{(i,j)\in\s{K \cup K' \setminus K \cap K' }}\mathbb{E}_{\calH_0}\pp{\frac{\calP(\s{X}_{ij})}{\calQ(\s{X}_{ij})}}\prod_{(i,j)\in\s{K\cap K'}}\mathbb{E}_{\calH_0}\pp{\frac{\calP(\s{X}_{ij})}{\calQ(\s{X}_{ij})}}^2}\label{eqn:73} \\
& \hspace{1cm}\stackrel{(a)}{=} \mathbb{E}_{\s{K}\indep\s{K}'\sim\s{Unif}(\calK_{k,m,n})}\pp{\p{\mathbb{E}_{\calH_0}\pp{\frac{\calP(\s{X}_{ij})}{\calQ(\s{X}_{ij})}}^2}^{\abs{\s{K} \cap \s{K}'}}}\label{eqn:74}\\
& \hspace{1cm}= \mathbb{E}_{\s{K}\indep\s{K}'\sim\s{Unif}(\calK_{k,m,n})}\pp{\p{1+\chi^2(\calP||\calQ)}^{\abs{\s{K} \cap \s{K}'}}}\\
&\hspace{1cm}\stackrel{(b)}{\leq}\mathbb{E}_{\s{K}\indep\s{K}'}\left[e^{\chi^2(\calP||\calQ)\cdot \abs{\s{K} \cap \s{K}'}}\right],
\end{align}
where $(a)$ is because $\bE_{\calQ}\frac{\calP(\s{X}_{ij})}{\calQ(\s{X}_{ij})}=1$, and $(b)$ is because $1+x\leq\exp(x)$, for any $x\in\mathbb{R}$.
\end{proof}

Based on Corollary~\ref{cor:1}, it suffices to upper bound $\mathbb{E}_{\s{K}\indep\s{K}'}\left[e^{\chi^2(\calP||\calQ)\cdot \abs{\s{K} \cap \s{K}'}}\right]$. Recall that $\s{K}$ and $\s{K}'$ are decomposed as $\s{K} = \bigcup_{\ell=1}^m\s{S}_\ell\times\s{T}_\ell$ and $\s{K}' = \bigcup_{\ell=1}^m\s{S}'_\ell\times\s{T}'_\ell$. Thus, we note that the intersection of $\s{K}$ and $\s{K}'$ can be rewritten as
\begin{align}
    \abs{\s{K} \cap \s{K}'} &= \sum_{\ell_1=1}^m\sum_{\ell_2=1}^m|(\s{S}_{\ell_1}\cap\s{S}'_{\ell_2})\times(\s{T}_{\ell_1}\cap\s{T}'_{\ell_2})| \\
    &= \sum_{\ell_1=1}^m\sum_{\ell_2=1}^m|(\s{S}_{\ell_1}\cap\s{S}'_{\ell_2})|\cdot|(\s{T}_{\ell_1}\cap\s{T}'_{\ell_2})|.
\end{align}
For each $\ell_1,\ell_2\in[m]$, define $\s{Z}_{\ell_1,\ell_2}\triangleq|(\s{S}_{\ell_1}\cap\s{S}'_{\ell_2})|$ and $\s{R}_{\ell_1,\ell_2}\triangleq|(\s{T}_{\ell_1}\cap\s{T}'_{\ell_2})|$. Note that the sequence of random variables $\{\s{Z}_{\ell_1,\ell_2}\}_{\ell_1,\ell_2}$ are statistically independent of the sequence $\{\s{R}_{\ell_1,\ell_2}\}_{\ell_1,\ell_2}$. Next, it is easy to show that $\s{Z}_{\ell_1,\ell_2}\sim\s{Hypergeometric}(n,k,k)$ and $\s{R}_{\ell_1,\ell_2}\sim\s{Hypergeometric}(n,k,k)$, for each $\ell_1,\ell_2\in[m]$, for any $\ell_1,\ell_2\in[m]$. Indeed, if we have an urn of $n$ balls among which $k$ balls are red, the random variable $\s{Z}_{\ell_1,\ell_2}$ (and $\s{R}_{\ell_1,\ell_2}$) is exactly the number of red balls if we draw $k$ balls from the urn uniformly at random without replacement, which is the definition of a Hypergeometric random variable. While the random variables $\{\s{Z}_{\ell_1,\ell_2}\}_{\ell_1,\ell_2}$ (and similarly $\{\s{R}_{\ell_1,\ell_2}\}_{\ell_1,\ell_2}$) are not independent, they are negatively associated. Thus,
\begin{align}
    \mathbb{E}_{\s{K}\indep\s{K}'}\left[e^{\chi^2(\calP||\calQ)\cdot \abs{\s{K} \cap \s{K}'}}\right]\leq \prod_{\ell_1=1}^m\prod_{\ell_2=1}^m\mathbb{E}\pp{e^{\chi^2(\calP||\calQ)\cdot \s{Z}_{\ell_1,\ell_2}\s{R}_{\ell_1,\ell_2}}} = \pp{\mathbb{E}\p{e^{\chi^2(\calP||\calQ)\cdot \s{Z}_{1,1}\s{R}_{1,1}}}}^{m^2}.
\end{align}
Next, it is well-known that $\s{Z}_{1,1} = \s{Hypergeometric}(n,k,k)$ (and similarly $\s{R}_{1,1} = \s{Hypergeometric}(n,k,k)$) is stochastically dominated by $
\s{B}\sim\s{Binomial}(k,k/n) = \sum_{i=1}^k\s{Bern}(k/n)$. Thus,
\begin{align}
    \mathbb{E}\p{e^{\chi^2(\calP||\calQ)\cdot \s{Z}_{1,1}\s{R}_{1,1}}}
    \leq \mathbb{E}\p{e^{\chi^2(\calP||\calQ)\cdot \s{B}\s{B}'}},
\end{align}
where $\s{B}'$ be an independent copy of $\s{B}$. 
Thus,
\begin{align}
   \mathbb{E}_{\s{K}\indep\s{K}'}\left[e^{\chi^2(\calP||\calQ)\cdot \abs{\s{K} \cap \s{K}'}}\right]\leq  \pp{\mathbb{E}\p{e^{\chi^2(\calP||\calQ)\cdot \s{B}\s{B}'}}}^{m^2}.\label{eqn:bound4}
\end{align}
We show that, if $\chi^2(\calP||\calQ)$ satisfies the condition of Theorem~\ref{thm:lower}, the  term on the right-hand side of \eqref{eqn:bound4} is at most $1+\delta$, for any $\delta>0$. 
We have
\begin{align}
    \pp{\mathbb{E}\p{e^{\chi^2(\calP||\calQ)\cdot \s{B}\s{B}'}}}^{m^2} = \pp{\bE\p{1+\frac{k}{n}\p{e^{\chi^2(\calP||\calQ)\s{B}}-1}}^k}^{m^2}.
\end{align}
Next, note that $\s{B}\leq k$ and we also assume the following, for reasons that will become clear,
\begin{align}
    \chi^2(\calP||\calQ)\leq \frac{1}{k}.
\end{align}
Therefore, using the inequality $e^x-1\leq x+ x^2$, for $x<1$, the following holds
\begin{align}
    \pp{\mathbb{E}\p{e^{\chi^2(\calP||\calQ)\cdot \s{B}\s{B}'}}}^{m^2} &\leq \pp{\bE\p{1+\frac{k}{n}\p{\chi^2(\calP||\calQ)\s{B}+\chi^4(\calP||\calQ)\s{B}^2}}^k}^{m^2}\\
    &\leq \pp{\bE\p{1+2\frac{k}{n}\chi^2(\calP||\calQ)\s{B}}^k}^{m^2}\\
    &\leq \pp{\bE\p{e^{2\frac{k^2}{n}\chi^2(\calP||\calQ)\s{B}}}}^{m^2}\\
    & = \pp{1+\frac{k}{n}\p{e^{2\frac{k^2}{n}\chi^2(\calP||\calQ)}-1}}^{km^2}.
\end{align}
This is at most $1+\delta$ if
\begin{align}
    \frac{k}{n}\p{e^{2\frac{k^2}{n}\chi^2(\calP||\calQ)}-1}\leq (1+\delta)^{\frac{1}{km^2}}-1.
\end{align}
Since $(1+\delta)^{\frac{1}{km^2}}-1\geq \log (1+\delta)/(km^2)$, this is implied by
\begin{align}
    \chi^2(\calP||\calQ)\leq \frac{n}{2k^2}\log\p{1+\frac{n\log (1+\delta)}{m^2k^2}}.
\end{align}
Putting altogether, we obtained that $\mathbb{E}_{\s{K}\indep\s{K}'}\left[e^{\chi^2(\calP||\calQ)\cdot \abs{\s{K} \cap \s{K}'}}\right]\leq 1+\delta$, if
\begin{align}
    \chi^2(\calP||\calQ)&\leq \min\ppp{\frac{1}{k},\frac{n}{2k^2}\log\p{1+\frac{n\log (1+\delta)}{m^2k^2}}}\\
    & = \min\ppp{\frac{1}{k},\frac{n^2\log (1+\delta)}{2m^2k^4}}.
\end{align}
Finally, note that in the Gaussian case, $\chi^2(\calN(\lambda,1)||\calN(0,1)) = \frac{1}{2}\pp{\exp\p{\lambda^2}-1}$. Thus, for $\lambda=o(1)$, we have $\chi^2(\calN(\lambda,1)||\calN(0,1))\to \frac{\lambda^2}{2}$, which concludes the proof.

\subsubsection{Consecutive submatrix detection}

For the consecutive case, we notice that by using the steps as in the previous subsection, we have
\begin{align}\mathbb{E}_{\calH_0}\left[\left(\mathsf{L}_n\left(\s{X}\right)\right)^2\right] \leq \mathbb{E}_{\s{K}\indep\s{K}'}\left[e^{\chi^2(\calP||\calQ)\cdot \abs{\s{K} \cap \s{K}'}}\right],
\end{align}
where $\s{K}$ and $\s{K}'$ are two independent copies drawn uniformly at random from $\calK_{k,m,n}^{\s{con}}$.
The key distinction from the previous case lies in the distribution of $|\s{K}\cap\s{K}'|$. Recall that $\s{K}$ and $\s{K}'$ are decomposed as $\s{K} = \bigcup_{\ell=1}^m\s{S}_\ell\times\s{T}_\ell$ and $\s{K}' = \bigcup_{\ell=1}^m\s{S}'_\ell\times\s{T}'_\ell$. Thus, we note that the intersection of $\s{K}$ and $\s{K}'$ can be rewritten as
\begin{align}
    \abs{\s{K} \cap \s{K}'} &= \sum_{\ell_1=1}^m\sum_{\ell_2=1}^m|(\s{S}_{\ell_1}\cap\s{S}'_{\ell_2})\times(\s{T}_{\ell_1}\cap\s{T}'_{\ell_2})| \\
    &= \sum_{\ell_1=1}^m\sum_{\ell_2=1}^m|(\s{S}_{\ell_1}\cap\s{S}'_{\ell_2})|\cdot|(\s{T}_{\ell_1}\cap\s{T}'_{\ell_2})|\\
    &\triangleq \sum_{\ell_1=1}^m\sum_{\ell_2=1}^m \s{Z}_{\ell_1,\ell_2}.
\end{align}
Note that for a given pair $(\ell_1,\ell_2)$, we have 
\begin{align}
    \pr(|(\s{S}_{\ell_1}\cap\s{S}'_{\ell_2})| = z) = \begin{cases}
        \frac{n-2k+1}{n},\ & \mathrm{for}\ z = 0\\
        \frac{2}{n},\ & \mathrm{for}\ z = 1,2,...,k-1\\
        \frac{1}{n},\ & \mathrm{for}\  z = k,
    \end{cases}\label{eqn:distZ}
\end{align}
and the exact same distribution for $|(\s{T}_{\ell_1}\cap\s{T}'_{\ell_2})|$. Thus, we may write $\s{Z}_{\ell_1,\ell_2}\stackrel{(d)}{=}\s{H}\cdot \s{H}'$, where $\s{H}$ and $\s{H}'$ are statistically independent and follow the distribution given in \eqref{eqn:distZ}. Thus, using the fact that the random variables $\ppp{\s{Z}_{\ell_1,\ell_2}}_{\ell_1,\ell_2}$ are negatively associated, we get,
\begin{align}
\mathbb{E}_{\s{K}\indep\s{K}'}\left[e^{\chi^2(\calP||\calQ)\cdot \abs{\s{K} \cap \s{K}'}}\right]\leq \prod_{\ell_1=1}^m\prod_{\ell_2=1}^m\mathbb{E}\pp{e^{\chi^2(\calP||\calQ)\cdot \s{Z}_{\ell_1,\ell_2}}} = \pp{\mathbb{E}\p{e^{\chi^2(\calP||\calQ)\cdot \s{H}\cdot \s{H}'}}}^{m^2}.\label{eqn:upperboundMomentsTildeK1}
\end{align}
Now,
\begin{align}
    \mathbb{E}\p{e^{\chi^2(\calP||\calQ)\cdot \s{H}\cdot \s{H}'}} & =\mathbb{E}\p{\frac{n-2k+1}{n} + \frac{2}{n}\sum_{i=1}^{k-1} e^{\chi^2(\calP||\calQ) \cdot i\s{H}'} + \frac{e^{\chi^2(\calP||\calQ) \cdot k\s{H}'}}{n}} \\
    & \leq \mathbb{E}\p{\frac{n-2k}{n} + \frac{2k}{n}e^{\chi^2(\calP||\calQ) \cdot k\s{H}'}} \\
    & = \frac{n-2k}{n}+\frac{2k}{n}\p{\frac{n-2k+1}{n} + \frac{2}{n}\sum_{i=1}^{k-1} e^{\chi^2(\calP||\calQ) \cdot ik} + \frac{e^{\chi^2(\calP||\calQ) \cdot k^2}}{n}}\\
    &\leq \frac{n-2k}{n}+\frac{2k}{n}\p{\frac{n-2k}{n} + \frac{2k}{n}e^{\chi^2(\calP||\calQ) \cdot k^2}}\\
    &= 1 + \frac{4k^2}{n^2} \p{e^{\chi^2(\calP||\calQ) \cdot k^2}-1}.
\end{align}
Therefore,
\begin{align}
\mathbb{E}_{\s{K}\indep\s{K}'}\left[e^{\chi^2(\calP||\calQ)\cdot \abs{\s{K} \cap \s{K}'}}\right]\leq  \pp{1 + \frac{4k^2}{n^2} \p{e^{\chi^2(\calP||\calQ) \cdot k^2}-1}}^{m^2}.\label{eqn:upperboundMomentsTildeK2}
\end{align}
This is at most $1+\delta$ if,
\begin{align}
    \frac{4k^2}{n^2}\p{e^{\chi^2(\calP||\calQ)k^2}-1}\leq (1+\delta)^{\frac{1}{m^2}}-1.
\end{align}
Since $(1+\delta)^{\frac{1}{m^2}}-1\geq \log (1+\delta)/(m^2)$, this is implied by
\begin{align}
    \chi^2(\calP||\calQ)\leq \frac{1}{k^2}\log\p{1+\frac{n^2\log (1+\delta)}{4k^2m^2}}.\label{eqn:GCDconc}
\end{align}
Finally, note that since $km\leq n$, the logarithmic factor in \eqref{eqn:GCDconc} can be lower bounded by $\log(1+\log(1+\delta)/4)$, which concludes the proof.

\subsection{Proof of Theorem~\ref{thm:gap}}\label{sec:StatComp}
In order to prove Theorem~\ref{thm:gap}, we use the following result \cite[Theorem 2.6]{Dmitriy19}.
\begin{lemma}\label{lemma:LowD}
    Let $\mathbf{S}$ be an $n$ dimensional random vector drawn from some distribution $\calD_n$, and let $\mathbf{Z}$ be an i.i.d. $n$ dimensional random vector with standard normal entries. Consider the detection problem:
\begin{align}
    \calH_0: \mathbf{Y} = \mathbf{Z} \quad\quad\s{vs.}\quad\quad\calH_1: \mathbf{Y} = \mathbf{S}+\mathbf{Z}.
\end{align}    
Then, 
\begin{align}
    \norm{\s{L}_n^{\leq \s{D}}}^2_{\calH_0} = \bE_{\mathbf{S}\indep\mathbf{S}'}\pp{\sum_{d=0}^{\s{D}}\frac{1}{d!}\innerP{\mathbf{S},\mathbf{S}'}^d},
\end{align}
where $\mathbf{S}$ and $\mathbf{S}'$ are drawn from $\calD_n$, and $\s{L}_n^{\leq D}$ is the $\s{D}$-low-degree likelihood ratio.
\end{lemma}

Our $\s{SD}$ problem falls under the setting of Lemma~\ref{lemma:LowD}. Specifically, let $\s{K}\sim\s{Unif}\pp{\calK_{k,m,n}}$, and define $\tilde{\mathbf{S}}$ to be an $n\times n$ matrix such that $[\tilde{\mathbf{S}}]_{ij} = \lambda$, if $i,j\in\s{K}$, and  $[\tilde{\mathbf{S}}]_{ij}=0$, otherwise. Also, we define $\mathbf{S}$ as the vectorized version of $\tilde{\mathbf{S}}$. Then, it is clear that our $\s{SD}$ problem cast as the detection problem in Lemma~\ref{lemma:LowD}, and thus,
\begin{align}
    \norm{\s{L}_n^{\leq \s{D}}}^2 &= \bE_{\mathbf{S}\indep\mathbf{S}'}\pp{\sum_{d=0}^{\s{D}}\frac{1}{d!}\innerP{\mathbf{S},\mathbf{S}'}^d}\\
    & = \sum_{d=0}^{\s{D}}\frac{\lambda^{2d}}{d!}\bE\abs{\s{K} \cap \s{K}'}^{d},
\end{align}
where we have used the fact that $\innerP{\mathbf{S},\mathbf{S}'} = \mathbf{S}^T\mathbf{S}' = \norm{\mathbf{S}\odot\mathbf{S}}_1 = \lambda^2\abs{\s{K} \cap \s{K}'}$, and $\s{K}'$ is an independent copy of $\s{K}$. Now, recall that $\s{K}$ and $\s{K}'$ are decomposed as $\s{K} = \bigcup_{\ell=1}^m\s{S}_\ell\times\s{T}_\ell$ and $\s{K}' = \bigcup_{\ell=1}^m\s{S}'_\ell\times\s{T}'_\ell$. Thus, we note that the intersection of $\s{K}$ and $\s{K}'$ can be rewritten as
\begin{align}
    \abs{\s{K} \cap \s{K}'} &= \sum_{\ell_1=1}^m\sum_{\ell_2=1}^m|(\s{S}_{\ell_1}\cap\s{S}'_{\ell_2})\times(\s{T}_{\ell_1}\cap\s{T}'_{\ell_2})| \\
    &= \sum_{\ell_1=1}^m\sum_{\ell_2=1}^m|(\s{S}_{\ell_1}\cap\s{S}'_{\ell_2})|\cdot|(\s{T}_{\ell_1}\cap\s{T}'_{\ell_2})|.
\end{align}
For each $\ell_1,\ell_2\in[m]$, define $\s{Z}_{\ell_1,\ell_2}\triangleq|(\s{S}_{\ell_1}\cap\s{S}'_{\ell_2})|$ and $\s{R}_{\ell_1,\ell_2}\triangleq|(\s{T}_{\ell_1}\cap\s{T}'_{\ell_2})|$. Recall from the previous subsection that the sequence of random variables $\{\s{Z}_{\ell_1,\ell_2}\}_{\ell_1,\ell_2}$ are statistically independent of the sequence $\{\s{R}_{\ell_1,\ell_2}\}_{\ell_1,\ell_2}$, and that $\s{Z}_{\ell_1,\ell_2}\sim\s{Hypergeometric}(n,k,k)$ and $\s{R}_{\ell_1,\ell_2}\sim\s{Hypergeometric}(n,k,k)$, for each $\ell_1,\ell_2\in[m]$, for any $\ell_1,\ell_2\in[m]$. Furthermore, $\{\s{Z}_{\ell_1,\ell_2}\}_{\ell_1,\ell_2}$ (and similarly $\{\s{R}_{\ell_1,\ell_2}\}_{\ell_1,\ell_2}$) are  negatively associated. Finally, recall that both $\s{Z}_{\ell_1,\ell_2}$ and $\s{R}_{\ell_1,\ell_2}$ are stochastically dominated by $\s{Binomial}(k,k/n)$. Thus, using \cite{Berend10} (see also \cite[Theorem 1]{Thomas22}), we have
\begin{align}
    \bE\abs{\s{K} \cap \s{K}'}^{d}\leq \s{B}_d^2\max\ppp{\frac{m^2k^4}{n^2},\p{\frac{m^2k^4}{n^2}}^d},
\end{align}
where $\s{B}_d$ is the $d$th Bell number. 
Thus, 
\begin{align}
    \norm{\s{L}_n^{\leq \s{D}}}^2 &\leq 1+ \sum_{d=1}^{\s{D}}\frac{\lambda^{2d}}{d!}\s{B}_d^2\max\ppp{\frac{m^2k^4}{n^2},\p{\frac{m^2k^4}{n^2}}^d} \triangleq 1+\sum_{d=1}^{\s{D}}\s{T}_d.
\end{align}
If $\frac{m^2k^4}{n^2}<1$, then it is clear that for $\sum_{d=1}^{\s{D}}\s{T}_d = O(1)$, it suffices that $\lambda<1$. On the other hand, if $\frac{m^2k^4}{n^2}>1$, then consider the ratio between successive terms:
\begin{align}
    \frac{\s{T}_{d+1}}{\s{T}_d} = \frac{\s{B}_{d+1}^2}{(d+1)\s{B}_{d}^2}\lambda^2m^2\frac{k^4}{n^2}.
\end{align}
Thus if $\lambda$ is small enough, namely if
\begin{align}
    \frac{mk^2\lambda}{n}\leq\frac{\sqrt{d+1}}{\sqrt{2}}\frac{\s{B}_d}{\s{B}_{d+1}},
\end{align}
then $\frac{\s{T}_{d+1}}{\s{T}_d}\leq1/2$, for all $1\leq d\leq \s{D}$. In this case, by comparing with a geometric sum, we may bound $\norm{\s{L}_n^{\leq \s{D}}}^2\leq O(1)$. This concludes the proof.

To show that the analysis above is tight, note that
\begin{align}
    \norm{\s{L}_n^{\leq \s{D}}}^2 &\sum_{d=0}^{\s{D}}\frac{\lambda^{2d}}{d!}\bE\abs{\s{K} \cap \s{K}'}^{d}\\
    &\geq \lambda^2\bE\abs{\s{K} \cap \s{K}'}\\
    & = \lambda^2m^2\frac{k^4}{n^2}.
\end{align}
Thus, if $\lambda$ is large enough, namely if $\lambda = \omega(n/(mk^2))$, then $\norm{\s{L}_n^{\leq \s{D}}}^2 = \omega(1)$.

\subsection{ML Estimator Derivation}\label{subsec:MLEDerivation}

The derivation below applies for both the case where $\s{K}\in\calK_{k,m,n}$ and the consecutive case where $\s{K}\in\calK_{k,m,n}^{\s{con}}$. Let $\pr_{\calH_1\vert\s{K}}(\s{X}\vert\s{K})$ denote the conditional distribution of $\s{X}$ given $\s{K}$. Recall that the ML estimate of $\s{K}$ is given by
\begin{align}
    \hat{\s{K}}_{\s{ML}}(\s{X}) = \argmax_{\s{K}\in\calK_{k,m,n}}\log\pr_{\calH_1\vert\s{K}}(\s{X}\vert\s{K}).\label{eqn:MLgeneral}
\end{align}
Given $\s{K}$, the distribution of $\s{X}$ under $\calH_1$ is given by, 
\begin{align}
    \log\pr_{\calH_1\vert\s{K}}(\s{X}\vert\s{K}) &= -\frac{n^2}{2}\log(2\pi e)-\frac{1}{2}\sum_{(i,j)\in\s{K}}(\s{X}_{ij}-\lambda)^2-\frac{1}{2}\sum_{(i,j)\not\in\s{K}}\s{X}^2_{ij}\\
    & = -\frac{n^2}{2}\log(2\pi e)+\lambda^2mk^2-\frac{1}{2}\sum_{(i,j)\in[n]^2}\s{X}^2_{ij}+\frac{\lambda}{2}\sum_{(i,j)\in\s{K}}\s{X}_{ij}.\label{eqn:MLderiva}
\end{align}
Noticing that only the last term at the r.h.s. of \eqref{eqn:MLderiva} depends on $\s{K}$, the ML estimator in \eqref{eqn:MLgeneral} boils down to
\begin{align}
    \hat{\s{K}}_{\s{ML}}(\s{X}) = \argmax_{\s{K}\in\calK_{k,m,n}}\sum_{(i,j)\in\s{K}}\s{X}_{ij}.\label{eqn:MLgeneral2}
\end{align}
For the consecutive model, the ML estimator is given by \eqref{eqn:MLgeneral2}, but with $\calK_{k,m,n}$ replaced by $\calK_{k,m,n}^{\s{con}}$. This problem maximizes the sum of entries among all $m$ principal submatrices of size $k\times k$ of $\s{X}$. 

\subsection{Proof of Theorem~\ref{thm:recupper}}

\subsubsection{Exact recovery using the ML estimator}\label{subsec:MLE}

In this subsection, we analyze the ML estimator. Recall that,  
\begin{align}
    \hat{\s{K}}_{\s{ML}}(\s{X}) = \argmax_{\s{K}\in\calK_{k,m,n}}\calS(\s{K}),
    \label{eq:MLProof}
\end{align}
where $\calS(\s{K})\triangleq\sum_{(i,j)\in\s{K}}\s{X}_{ij}$. 
We next prove the conditions for which $\hat{\s{K}}_{\s{ML}} = \s{K}^\star$, with high probability, where $\s{K}^\star$ are the $m$ planted submatrices. To prove the theorem, it suffices to show that $\calS(\s{K}^\star)>\calS(\s{K})$, for all feasible $\s{K}$ with $\s{K}\neq \s{K}^\star$. Let $\s{D}(\s{K})\triangleq\calS(\s{K}^\star)-\calS(\s{K})$. 
Note that
\begin{align}
    \s{D}(\s{K}) &= \sum_{(i,j)\in\s{K}^\star}\s{X}_{ij}-\sum_{(i,j)\in\s{K}}\s{X}_{ij}\\
    & = \sum_{(i,j)\in\s{K}^\star}\bE\s{X}_{ij}-\sum_{(i,j)\in\s{K}}\bE\s{X}_{ij}+\sum_{(i,j)\in\s{K}^\star}[\s{X}_{ij}-\bE\s{X}_{ij}]-\sum_{(i,j)\in\s{K}}[\s{X}_{ij}-\bE\s{X}_{ij}]\\
    & = \lambda\cdot(mk^2-\abs{\s{K}^\star\cap\s{K}})+\sum_{(i,j)\in\s{K}^\star\setminus\s{K}}[\s{X}_{ij}-\lambda]-\sum_{(i,j)\in\s{K}\setminus\s{K}^\star}\s{X}_{ij}\\
    & = \lambda\cdot(mk^2-\abs{\s{K}^\star\cap\s{K}})+\s{W}_1(\s{K})+\s{W}_2(\s{K}),\label{eqn:DKexpansionMLE}
\end{align}
where $\s{W}_1(\s{K})\triangleq\sum_{(i,j)\in\s{K}^\star\setminus\s{K}}[\s{X}_{ij}-\lambda]$ and $\s{W}_2(\s{K})\triangleq-\sum_{(i,j)\in\s{K}\setminus\s{K}^\star}\s{X}_{ij}$.
Because $|\s{K}| = |\s{K}^\star|=mk^2$, we have $|\s{K}^\star\setminus\s{K}|=|\s{K}\setminus\s{K}^\star| = mk^2-\abs{\s{K}^\star\cap\s{K}}$.
Thus, both $\s{W}_1(\s{K})$ and $\s{W}_2(\s{K})$ are composed of the sum of $mk^2-\abs{\s{K}^\star\cap\s{K}}$ i.i.d. centered Gaussian random variables with unit variance. Accordingly, for $i=1,2$, and each fixed $\s{K}$,
\begin{align}
    \pr \p{\s{W}_i(\s{K}) \leq -\frac{\lambda}{2}(mk^2-|\s{K}^\star\cap\s{K}|)} \leq \frac{1}{2}\exp\pp{-\frac{1}{2}\lambda^2(mk^2-|\s{K}^\star\cap\s{K}|)},
\end{align}
and therefore, by the union bound and \eqref{eqn:DKexpansionMLE},
\begin{align}
    \pr\p{\s{D}(\s{K})\leq0}\leq \exp\pp{-\frac{1}{2}\lambda^2(mk^2-|\s{K}^\star\cap\s{K}|)}.\label{eqn:COncenDMLE}
\end{align}
Using \eqref{eqn:COncenDMLE} and the union bound once again, we get
\begin{align}
    \pr\p{\hat{\s{K}}_{\s{ML}}(\s{X})\neq\s{K}^\star} &= \pr\pp{\bigcup_{\s{K}\neq\s{K}^\star}\s{D}(\s{K})\leq0}\\
    &\leq \sum_{\s{K}\neq\s{K}^\star}\pr\p{\s{D}(\s{K})\leq0}\\
    &\leq \sum_{\s{K}\neq\s{K}^\star}\exp\pp{-\frac{1}{2}\lambda^2(mk^2-|\s{K}^\star\cap\s{K}|)}\\
    &= \sum_{\ell=0}^{mk^2-k}\abs{\s{K}\in\calK_{k,m, n}^{\s{con}}:|\s{K}^\star\cap\s{K}|=\ell}e^{-\frac{1}{2}\lambda^2(mk^2-\ell)},
\end{align}
where the last equality follows from the fact that since $\s{K}^\star,\s{K}\in\calK_{k,m, n}^{\s{con}}$ and $\s{K}^\star\cap\s{K}\neq\emptyset$, we must have that $|\s{\s{K}^\star}\cap\s{K}|\leq mk^2-k$. It can be shown that $\abs{\s{K}\in\calK_{k,m, n}^{\s{con}}:|\s{K}^\star\cap\s{K}|=\ell}\leq C\frac{(mk^2-\ell)^2}{k^2}n^{\frac{C'(mk^2-\ell)}{k}}$, from some $C,C'>0$, see, e.g., \cite[Lemma 7]{chen2016statistical}. Then,
\begin{align}
    \pr\p{\hat{\s{K}}_{\s{ML}}(\s{X})\neq\s{K}^\star}&\leq C\sum_{\ell=0}^{mk^2-k}\frac{(mk^2-\ell)^2}{k^2}n^{\frac{C'(mk^2-\ell)}{k}}e^{-\frac{1}{2}\lambda^2(mk^2-\ell)}\\
    & = C\sum_{\ell=k}^{mk^2}\frac{\ell^2}{k^2}n^{\frac{C'\ell}{k}}e^{-\frac{1}{2}\lambda^2\ell}\\
    &\leq C\sum_{\ell=k}^{mk^2}n^4n^{\frac{C'\ell}{k}}e^{-\frac{1}{2}\lambda^2\ell}\\
    & = Cn^4\sum_{\ell=k}^{mk^2}e^{\frac{C'\ell}{k}\log n-\frac{1}{2}\lambda^2\ell}\\
    & = Cn^4\sum_{\ell=k}^{mk^2}e^{-\ell\cdot\p{\frac{1}{2}\lambda^2-\frac{C'}{k}\log n}}\\
    &\stackrel{(a)}{\leq} Cn^4\sum_{\ell=k}^{mk^2}e^{-\frac{8\ell}{k}\log n}\\
    &\leq \frac{Cn^4mk^2}{n^8} = C\frac{mk^2}{n^4},
    \end{align}
where in (a) we have used the fact that $\lambda^2>\frac{(2C'+16)\log n}{k}$. Thus, we get that $\pr(\hat{\s{K}}_{\s{ML}}(\s{X})\neq\s{K}^\star)$ converges to zero, as $n\to\infty$.

\subsubsection{Exact recovery using the peeling estimator}\label{subsec:ML0}

We analyze the first step of the peeling algorithm (which boils down to the ML estimator for a single planted submatrix), and the strategy to bound each of the other sequential steps is exactly the same. Recall that, 
\begin{align}
    \hat{\s{K}}_{1}(\s{X}) = \argmax _{\s{K}\in\calK_{k,1, n}^{\s{con}}} \calS(\s{K}),
    \label{eq:exhaustive}
\end{align}
where $\calS(\s{K})\triangleq\sum_{(i,j)\in\s{K}}\s{X}_{ij}$. We next prove the conditions for which $\hat{\s{K}}_{1}(\s{X}) = \s{K}^\star_\ell$, with high probability, for some $\ell\in[m]$, where $\s{K}^\star = \cup_{\ell=1}^m\s{K}_\ell^\star$ are the $m$ planted submatrices. To prove the theorem it suffices to show that $\calS(\s{K})>\max_{\ell\in[m]}\calS(\s{K}^\star_\ell)$ is asymptotically small, for all feasible $\s{K}$ with $\s{K} \neq \s{K}_\ell^\star$, for $\ell\in[m]$. Let $\s{D}_\ell(\s{K})\triangleq\calS(\s{K}^\star_\ell)-\calS(\s{K})$.  
Note that
\begin{align}
    \s{D}_\ell(\s{K}) &= \sum_{(i,j)\in\s{K}_\ell^\star}\s{X}_{ij}-\sum_{(i,j)\in\s{K}}\s{X}_{ij}\\
    & = \sum_{(i,j)\in\s{K}_\ell^\star}\bE\s{X}_{ij}-\sum_{(i,j)\in\s{K}}\bE\s{X}_{ij}+\sum_{(i,j)\in\s{K}_\ell^\star}[\s{X}_{ij}-\bE\s{X}_{ij}]-\sum_{(i,j)\in\s{K}}[\s{X}_{ij}-\bE\s{X}_{ij}]\\
    & = \lambda\cdot(k^2-\abs{\s{K}^\star\cap\s{K}})+\sum_{(i,j)\in\s{K}_\ell^\star\setminus\s{K}}[\s{X}_{ij}-\lambda]-\sum_{(i,j)\in\s{K}\setminus\s{K}_\ell^\star}[\s{X}_{ij}-\bE\s{X}_{ij}]\\
    & = \lambda\cdot(k^2-\abs{\s{K}^\star\cap\s{K}})+\s{W}_1(\s{K})+\s{W}_2(\s{K}),\label{eqn:DKexpansion}
\end{align}
where $\s{W}_1(\s{K})\triangleq\sum_{(i,j)\in\s{K}_\ell^\star\setminus\s{K}}[\s{X}_{ij}-\lambda]$ and $\s{W}_2(\s{K})\triangleq-\sum_{(i,j)\in\s{K}\setminus\s{K}_\ell^\star}[\s{X}_{ij}-\bE\s{X}_{ij}]$.
Because $|\s{K}| = |\s{K}_\ell^\star|=k^2$, we have $|\s{K}_\ell^\star\setminus\s{K}|=|\s{K}\setminus\s{K}_\ell^\star| = k^2-\abs{\s{K}_\ell^\star\cap\s{K}}$.
Thus, both $\s{W}_1(\s{K})$ and $\s{W}_2(\s{K})$ are composed of sum of $k^2-\abs{\s{K}_\ell^\star\cap\s{K}}$ i.i.d. centered Gaussian random variables with unit variance. Accordingly, for $i=1,2$, and each fixed $\s{K}$,
\begin{align}
    \pr \p{\s{W}_i(\s{K}) \leq -\frac{\lambda}{2}(k^2-|\s{K}^\star\cap\s{K}|)} &\leq \frac{1}{2}\exp\pp{-\frac{\lambda^2}{8}\frac{(k^2-|\s{K}^\star\cap\s{K}|)^2}{k^2-\abs{\s{K}_\ell^\star\cap\s{K}}}}.
\end{align}
Therefore, by the union bound and \eqref{eqn:DKexpansion},
\begin{align}
    \pr\p{\s{D}_\ell(\s{K})\leq0}\leq \exp\pp{-\frac{\lambda^2}{8}\frac{(k^2-|\s{K}^\star\cap\s{K}|)^2}{k^2-\abs{\s{K}_\ell^\star\cap\s{K}}}}.\label{eqn:COncenD}
\end{align}
Note that due to the separation assumption, it must be the case that either $|\s{K}^\star\cap\s{K}| = |\s{K}_j^\star\cap\s{K}|\neq0$, for some $j\in[m]$, or $|\s{K}^\star\cap\s{K}|=0$. In the later case, we have
\begin{align}
    \pr\p{\s{D}_\ell(\s{K})\leq0}\leq \exp\pp{-\frac{\lambda^2k^2}{8}},
\end{align}
while in the former the exists a unique $j\in[m]$, such that,
\begin{align}
     \min_{\ell\in[m]}\pr\p{\s{D}_\ell(\s{K})\leq0}&\leq \min_{\ell\in[m]}\exp\pp{-\frac{\lambda^2}{8}\frac{(k^2-|\s{K}_j^\star\cap\s{K}|)^2}{k^2-\abs{\s{K}_\ell^\star\cap\s{K}}}}\\
     &\leq \exp\pp{-\frac{\lambda^2}{8}(k^2-|\s{K}_j^\star\cap\s{K}|)}\\
     &\leq \exp\pp{-\frac{\lambda^2k}{8}},
\end{align}
where the third inequity is since $\s{K}_j^\star,\s{K}\in\calK_{k,1, n}^{\s{con}}$ and $\s{K}_j^\star\cap\s{K}\neq\emptyset$, we must have that $|\s{\s{K}_j^\star}\cap\s{K}|\leq k^2-k$. Accordingly, using \eqref{eqn:COncenD} and the union bound once again, we get
\begin{align}
    \pr\p{\hat{\s{K}}_{1}(\s{X})\neq\s{K}_\ell^\star\;\text{for some }\ell\in[m]} &= \pr\pp{\bigcup_{\s{K}\neq(\s{K}_1^\star,\ldots,\s{K}_m^\star)}\ppp{\calS(\s{K})>\max_{\ell\in[m]}\calS(\s{K}^\star_\ell)}}\\
    & =\pr\pp{\bigcup_{\s{K}\neq(\s{K}_1^\star,\ldots,\s{K}_m^\star)}\ppp{\s{D}_1(\s{K})\leq0,\ldots,\s{D}_m(\s{K})\leq0}}\\
    &\leq \sum_{\s{K}\neq(\s{K}_1^\star,\ldots,\s{K}_m^\star)}\min_{\ell\in[m]}\pr\p{\s{D}_\ell(\s{K})\leq0}\\
    &\leq \sum_{\s{K}\neq(\s{K}_1^\star,\ldots,\s{K}_m^\star)}\exp\pp{-\frac{\lambda^2k}{8}}\\
    \leq n^2e^{-\frac{1}{8}\lambda^2k},
\end{align}
where the last inequality is because $|\calK_{k,1, n}^{\s{con}}|\leq n^2$. Thus, we see that if $\lambda^2>\frac{(24+\epsilon)\log n}{k}$, then $\pr\p{\hat{\s{K}}_{1}(\s{X})\neq\s{K}_\ell^\star\;\text{for some }\ell\in[m]}\leq n^{-(1+\epsilon/8)}$. Using the same steps above, it is clear that, $\pr(\hat{\s{K}}_{\ell}(\s{X})\neq\s{K}_\ell^\star)\leq n^{-(1+\epsilon/8)}$, for any $2\leq\ell\leq m$, provided that $\lambda^2>\frac{(24+\epsilon)\log n}{k}$. Thus, 
\begin{align}
    \pr\pp{\hat{\s{K}}_{\s{peel}}\neq\s{K}^\star} = \pr\pp{\bigcup_{\ell=1}^m\hat{\s{K}}_{\ell}\neq\s{K}_\ell^\star}\leq \frac{m}{n^{(1+\epsilon/8)}} = n^{-\epsilon/8},
\end{align}
which converges to zero as $n\to\infty$.

\subsubsection{Correlated recovery using the peeling estimator}

Our analysis for correlated recovery relies on standard arguments as in \cite{banks2018information,wu2018statistical}. Recall the peeling estimator in \eqref{eqn:MLPeel}. Denote the planted submatrices by $\s{K}^\star = \bigcup_{i=1}^m\s{T}^\star_i\times\s{S}_i^\star\in\calK^{\s{con}}_{k,m,n}$. We let $\s{K}_i^\star\triangleq \s{T}^\star_i\times\s{S}_i^\star$, for $i\in[m]$, and fix $\epsilon>0$. 
Let us analyze the first step of the algorithm, i.e., $\hat{\s{K}}_{1}(\s{X}^{(1)})=\hat{\s{K}}_{1}(\s{X})$. Recall that
\begin{align}
   \hat{\s{K}}_{1}(\s{X}^{(1)}) = \argmax_{\s{K}\in\calK_{k,1,n}^{\s{con}}}\calS_1(\s{K}),
\end{align}
where we define $\calS_\ell(\s{K})\triangleq\sum_{(i,j)\in\s{K}}\s{X}_{ij}^{(\ell)}$, for $\s{K}\in\calK_{k,1,n}^{\s{con}}$ and $\ell\in[m]$. Under the planted model, $\calS_1(\s{K})\sim\calN(\lambda\langle{\s{K},\s{K}^\star}\rangle,k^2)$. Hence, the distribution of $\calS_1(\s{K})$ depends on the size of the overlap of $\s{K}$ with $\s{K}^\star$. To prove that reconstruction is possible, we compute in the planted model the probability that $\calS_1(\s{K})>\max_{\ell\in[m]}\calS_1(\s{K}_\ell^\star)$, given that $\s{K}$ has overlap $\langle{\s{K},\s{K}^\star}\rangle=\omega$ with the planted partition, and argue that this probability tends to zero whenever the overlap is small enough. For each $\ell\in[m]$, note that $\calS_1(\s{K}_\ell^\star)\sim\calN(\lambda k^2,k^2)$, and thus Hoeffding's inequality implies that $\calS_1(\s{K}_\ell^\star)> \lambda k^2-\sqrt{2k^2\log n}$, with probability at least $1-O(n^{-1})$. Taking the union bound over every $\s{K}$ with overlap at most $\omega$ gives
\begin{align}
&\pr\p{\max_{\langle{\s{K},\s{K}^\star}\rangle\leq\omega}\calS_1(\s{K})>\lambda k^2-\sqrt{2k^2\log n}}\leq \\
& \hspace{0.8cm}= n^2\cdot \max_{\langle{\s{K},\s{K}^\star}\rangle\leq\omega}\exp\p{-\frac{\pp{\lambda (k^2-\langle{\s{K},\s{K}^\star}\rangle)-\sqrt{2k^2\log n}}^2}{2k^2}}\\
    &\hspace{0.8cm} = n^2\cdot\max_{\langle{\s{K},\s{K}^\star}\rangle\leq\omega}\exp\p{-\pp{\frac{\lambda}{\sqrt{2k^2}} (k^2-\langle{\s{K},\s{K}^\star}\rangle)-\sqrt{\log n}}^2}\\
    &\hspace{0.8cm}\leq\exp\p{2\log n-\pp{\frac{\lambda}{\sqrt{2k^2}} (k^2-\omega)-\sqrt{\log n}}^2}.
\end{align}
By the assumption that $\lambda>\frac{\s{C}\sqrt{\log n}}{k}$, with $\s{C}>2+\sqrt{2}$, it follows that there exists a fixed constant $\epsilon>0$ such that $(1-\epsilon)\lambda>\frac{\s{C}\sqrt{\log n}}{k}$. Hence, setting $\omega = k^2\epsilon$ in the last displayed equation, we get
\begin{align}
&\pr\p{\max_{\langle{\s{K},\s{K}^\star}\rangle\leq\omega}\calS_1(\s{K})>\lambda k^2-\sqrt{2k^2\log n}}\nonumber\\
    &\hspace{0.8cm}\leq \exp\p{2\log n-\pp{\frac{\lambda\sqrt{k^2}}{\sqrt{2}}(1-\epsilon)-\sqrt{\log n}}^2} = e^{-\Omega(n)},
\end{align}
and thus, with probability at least $1-e^{-\Omega(n)}$,
\begin{align}
    \max_{\langle{\s{K},\s{K}^\star}\rangle\leq k^2\epsilon}\calS_1(\s{K})<\lambda k^2-\sqrt{2k^2\log n}.
\end{align}
Consequently, we get that the maximum likelihood estimator $\hat{\s{K}}_{1}(\s{X}^{(1)})$ in \eqref{eqn:MLPeel} satisfies $\langle{\hat{\s{K}}_{1},\s{K}^\star}\rangle\geq k^2\epsilon$ with high probability. Finally, the separation assumption implies that there exist a unique $j\in[m]$ such that $\langle{\hat{\s{K}}_{1},\s{K}^\star}\rangle = \langle{\hat{\s{K}}_{1},\s{K}_j^\star}\rangle\geq k^2\epsilon$, and $ \langle{\hat{\s{K}}_{1},\s{K}_\ell^\star}\rangle=0$, for $\ell\neq j$. Then, in the second step of the peeling algorithm, we first compute $\s{X}_{ij}^{(2)}$, by setting $[\s{X}_{ij}^{(2)}]_{ij}=-\infty$, for any $(i,j)\in\hat{\s{K}}_{1}$, and $[\s{X}_{ij}^{(2)}]_{ij}=0$, otherwise. Thus, it is clear that in the second step,  $\hat{\s{K}}_{2}$ cannot be attained by any set that is $k$-closed to $\hat{\s{K}}_{1}$; indeed, $\calS_2(\s{K})=-\infty$, for any set $\s{K}$ that is $k$-closed to $\hat{\s{K}}_{1}$. Therefore, for the relevant sets in the maximization in $\hat{\s{K}}_{2}$, we again have $\calS_2(\s{K})\sim\calN(\lambda\langle{\s{K},\s{K}^\star}\rangle,k^2)$. Accordingly, repeating the exact same arguments above, we obtain that $\langle{\hat{\s{K}}_{2},\s{K}_j^\star}\rangle\geq k^2\epsilon$ with high probability, for some $j\in[m]$. In the same way, we get that $\langle{\hat{\s{K}}_{\ell},\s{K}_\ell^\star}\rangle\geq k^2\epsilon$ with high probability, for any $\ell\in[m]$. The union bound, then implies that
\begin{align}  \pr\pp{\bigcup_{\ell=1}^m\ppp{\langle{\hat{\s{K}}_{\ell},\s{K}_\ell^\star}\rangle< k^2\epsilon}}\leq m\cdot\pr\pp{\langle{\hat{\s{K}}_{\ell},\s{K}_\ell^\star}\rangle< k^2\epsilon}\leq me^{-\Omega(n)} = o(1).
\end{align}
Thus, $\langle{\hat{\s{K}}_{\s{peel}},\s{K}^\star}\rangle\geq mk^2\epsilon$ with high probability, namely, $\hat{\s{K}}_{\s{peel}}$ achieves correlated recovery.

\subsection{Proof of Theorem~\ref{thm:recLower}}\label{eqn:RecoveryLower}

\subsubsection{Exact recovery}

We use an information theoretical argument via Fano's inequality. Recall that $\calK_{k,m,n}^{\s{con}}$ is the set of possible planted submatrices. Let $\bar{\calK}_{k,m,n}$ be a subset of $\calK_{k,m,n}^{\s{con}}$, which will be specified
later on. Let $\bar{\pr}_{\s{X},\s{K}^\star}$ denote the joint distribution of the underlying location of the planted submatrices $\s{K}^\star$ and $\s{X}$, when $\s{K}^\star$ is drawn uniformly at random from $\bar{\calK}_{k,m,n}$, and $\s{X}$ is generated according to Definition~\ref{def:gconsrec}. Let $I(\s{X};\s{K}^\star)$ denote the mutual information between $\s{X}$ and $\s{K}^\star$. Then, Fano's inequality implies that,
\begin{align}
    \inf_{\hat{\s{K}}} \sup_{\s{K}^\star\in\bar{\calK}_{k,m,n}} \pr \pp{\hat{\s{K}} \neq \s{K}^\star} \geq \inf_{\hat{\s{K}}} \bar{\pr}\pp{\hat{\s{K}} \neq \s{K}^\star} \geq 1 - \frac{I(\s{X};\s{K}^\star)+1}{\log|\bar{\calK}_{k,m,n}|}.\label{eqn:Fano}
\end{align}
We construct $\bar{\calK}_{k,m,n}$ as follows. Let $\s{M}\triangleq \alpha\cdot m$, where $\alpha\in\mathbb{N}$ will be specified later on, and $\bar{\calK}_{k,m,n} = \ppp{\s{K}_\ell}_{\ell=0}^{\s{M}}$, where:
\begin{enumerate}
    \item The \emph{base} submatrix $\s{K}_0$ is defined as $\s{K}_0 = \bigcup_{\ell=1}^m\s{S}_\ell^0 \times \s{T}_\ell^0$, with $\s{S}_\ell^0=\ppp{(\ell-1)\cdot (k+\alpha)+1,\ldots,(\ell-1)\cdot (k+\alpha)+k}$ and $\s{T}_\ell^0=[k]$, for $\ell\in[m]$. Namely, every pair of consecutive matrices among the $m$ matrices in $\s{K}_0$ are $\alpha$ columns far apart.
    \item We let $\s{K}_{(j-1)\alpha+i}$, for $j=1,2,\ldots,m$ and $i = 1,2,\ldots,\alpha$, to be defined the same as $\s{K}_{0}$ but with $\s{S}^{j-1}$ shifted $i$ columns to the right.
\end{enumerate}
Let $\bar{\pr}_i$ denote the conditional distribution of $\s{X}$ given $\s{K}^\star = \s{K}_i$. Note that,
\begin{align}
    I(\s{X};\s{K}^\star) &= d_{\s{KL}}(\bar{\pr}_{\s{X},\s{K}^\star}||\bar{\pr}_{\s{X}}\bar{\pr}_{\s{K}^\star})\\
    & = \bE_{\s{K}^\star}\pp{d_{\s{KL}}(\bar{\pr}_{\s{X}\vert\s{K}^\star}||\bar{\pr}_{\s{X}})}\\
    & = \frac{1}{\s{M+1}}\sum_{i=0}^{\s{M}}d_{\s{KL}}(\bar{\pr}_i||\bar{\pr}_{\s{X}})\\
    &\leq\frac{1}{(\s{M+1})^2}\sum_{i,j=0}^{\s{M}}d_{\s{KL}}(\bar{\pr}_i||\bar{\pr}_j),
\end{align}
where the inequality follows from the fact that $\bar{\pr}_{\s{X}} = \frac{1}{\s{M+1}}\sum_{j=0}^{\s{M}}\bar{\pr}_j$, and the convexity of KL divergence. Now, since each $\bar{\pr}_i$ is a product of $n^2$ Gaussian distributions, we get
\begin{align}
    I(\s{X};\s{K}^\star) &\leq \frac{1}{(\s{M}+1)^2}\sum_{i, j=0}^{\s{M}} d_{\s{KL}}(\bar{\pr}_i||\bar{\pr}_j) \\ 
    &=\frac{1}{(\s{M}+1)}\sum_{j=0}^{\s{M}} d_{\s{KL}}(\bar{\pr}_0||\bar{\pr}_j) \\ 
    & =\frac{2km}{(\s{M}+1)} \pp{d_{\s{KL}}\p{\cN(\lambda,1)||\cN(0,1)}+ d_{\s{KL}}\p{\cN(0,1)||\cN(\lambda,1)}} \sum_{j=0}^{\alpha} j \\
    & = \frac{2km}{(\s{M}+1)} \frac{\alpha(1+\alpha)}{2} \pp{d_{\s{KL}}\p{\cN(\lambda,1)||\cN(0,1)}+ d_{\s{KL}}\p{\cN(0,1)||\cN(\lambda,1)}}\\
    &\leq 2k\frac{(1+\alpha)}{2} \pp{d_{\s{KL}}\p{\cN(\lambda,1)||\cN(0,1)}+ d_{\s{KL}}\p{\cN(0,1)||\cN(\lambda,1)}}\\
    &= (1+\alpha)k\lambda^2.
\end{align}
Thus, substituting the last inequality in \eqref{eqn:Fano}, and using the fact that $|\bar{\calK}_{k,m,n}| = 1+\s{M}$, we get that $\inf_{\hat{\s{K}}} \sup_{\s{K}^\star\in\bar{\calK}_{k,m,n}} \pr \pp{\hat{\s{K}} \neq \s{K}^\star}>1/2$, if
\begin{align}
    \lambda^2<\frac{\frac{1}{2}\log(1+\s{M})-1}{(1+\alpha)k} = \frac{\frac{1}{2}\log{(1+\alpha m)}-1}{(1+\alpha)k}.
\end{align}
Finally, it is clear that there exists $\alpha_0\in\mathbb{N}$, such that for any $\alpha>\alpha_0$, the minimax error probability is at least half, if $\lambda^2<C/k$, for some constant $C>0$, which concludes the proof.

\subsubsection{Correlated recovery}

The correlated recovery lower bound follows almost directly from the same arguments as in, e.g., \cite[Subsection 3.1.3]{wu2018statistical}. For completeness, we present here the main ideas in the proof. Note that the observations can be written as $\s{X} = \lambda\s{M}+\s{W}$, $\s{W}$ is an $n\times n$ i.i.d. matrix with zero mean and unit variance Gaussian entries, and $\s{M}$ is an $n\times n$ binary matrix such that $\s{M}_{ij}=1$ if $(i,j)\in\s{K}$, and $\s{M}_{ij}=0$, otherwise, and $\s{K}$ is the planted set. Define $\s{A} = \beta \lambda\s{M}+\s{W}$, where $\beta\in[0,1]$. The minimum mean-squared error estimator (MMSE) of $\s{M}$ given $\s{A}$ is $\hat{\s{M}}_{\s{MMSE}} = \bE\pp{\s{M}\vert\s{A}}$, and the rescaled minimum mean-squared error is $\s{MMSE}(\beta) = \frac{1}{mk^2}\bE\norm{\s{M}-\bE\pp{\s{M}\vert\s{A}}}^2_{\s{F}}$. Note that under the conditions of Theorem~\ref{thm:recLower}, we proved in Theorem~\ref{thm:lower} that $\chi^2(\calP||\calQ)<C$, from some constant $C>0$. Jensen's inequality implies that the KL divergence between $\calP$ and $\calQ$ is also bounded, indeed,
\begin{align}
    d_{\s{KL}}(\calP||\calQ)\leq\log \bE_{\calP}\frac{\calP}{\calQ}\leq \log C.
\end{align}
The main idea in the proof is to show that bounded KL divergence implies that for all $\beta\in[0,1]$, the MMSE tends to that of the trivial estimator $\hat{\s{M}}=0$, i.e.,
\begin{align}
    \lim_{n\to\infty}\s{MMSE}(\beta) = \lim_{n\to\infty}\frac{1}{mk^2}\bE\norm{\s{M}}^2_{\s{F}} = \lambda^2.\label{eqn:holdMMSE}
\end{align}
Expanding the MMSE in the left-hand-side of \eqref{eqn:holdMMSE}, we get
\begin{align}
    \lim_{n\to\infty}\frac{1}{mk^2}\bE\pp{-2\innerP{\s{M},\bE\pp{\s{M}\vert\s{A}}}+\norm{\bE\pp{\s{M}\vert\s{A}}}_{\s{F}}^2} =0,
\end{align}
which by the tower property of conditional expectation implies that
\begin{align}
   \lim_{n\to\infty}\frac{1}{mk^2}\bE\norm{\bE\pp{\s{M}\vert\s{A}}}_{\s{F}}^2 =0.\label{eqn:littlemk2}
\end{align}
Thus, the optimal estimator converges to the trivial one. To prove \eqref{eqn:holdMMSE}, a straightforward calculation shows that the mutual information between $\s{A}$ and $\s{M}$ is given by $I(\beta) = I(\s{M};\s{A}) = -d_{\s{KL}}(\calP||\calQ)+\frac{\beta}{4}\bE\norm{\s{M}}^2_{\s{F}}$. Thus, under the conditions of Theorem~\ref{thm:recLower},
\begin{align}
    \lim_{n\to\infty}\frac{1}{mk^2}I(\beta) = \frac{\beta\lambda^2}{4}.
\end{align}
Then, using the above and the I-MMSE formula \cite{GuoShamaiVerdu} it can be shown that \eqref{eqn:holdMMSE} holds true (see, \cite[eqns. (13)--(15)]{wu2018statistical}). Next, for any estimator $\hat{\s{K}}$ of the planted set, we can define an estimator for $\s{M}$ by $\hat{\s{M}}_{ij} = 1$ if $(i,j)\in\hat{\s{K}}$, and $\hat{\s{M}}_{ij} = 0$, otherwise. Then, using the Cauchy-Schwarz inequality, we have
\begin{align}
    \bE\langle{\s{M},\hat{\s{M}}}\rangle &= \bE\langle{\bE\pp{\s{M}\vert\s{A}},\hat{\s{M}}}\rangle\\
    &\leq \bE\pp{\norm{\bE\pp{\s{M}\vert\s{A}}}_{\s{F}}||\hat{\s{M}}||_{\s{F}}}\\
    &\leq \sqrt{\bE\norm{\bE\pp{\s{M}\vert\s{A}}}_{\s{F}}^2}\lambda\sqrt{mk^2} = o(mk^2),\label{eqn:smallinner}
\end{align}
where the last transition follows from \eqref{eqn:littlemk2}. Thus, \eqref{eqn:smallinner} implies that for any estimator $\hat{\s{K}}$, we have $\bE\langle{\s{K},\hat{\s{K}}}\rangle = o(mk^2)$, and thus correlated recovery of $\s{K}$ is impossible.


\section{Conclusions and future work}\label{sec:Conc}

In this paper, we study the computational and statistical boundaries of the submatrix and consecutive submatrix detection and recovery problems. For both models, we derive asymptotically tight lower and upper bounds on the thresholds for detection and recovery. To that end, for each problem, we propose statistically optimal and efficient algorithms for detection and recovery and analyze their performance. Our statistical lower bounds are based on classical techniques from information theory. Finally, we use the framework of low-degree polynomials to provide evidence of the statistical-computational gap we observed in the submatrix detection problem. 

There are several exciting directions for future work. First, it would be interesting to generalize our results to any pair of distributions $\calP$ and $\calQ$. While our information-theoretic lower bounds hold for general distributions, it is left to construct and analyze algorithms for this case, as well as to derive computational lower bounds. In our paper, we assume that the elements inside the planted submatrices are i.i.d., however, it is of practical interest to generalize this assumption and consider the case of dependent entries, e.g., Gaussians with a general covariance matrix. For example, this is the typical statistical model of cryo-EM data~\cite{bendory2020single}. Finally, it will be interesting to prove a computational lower bound for the submatrix recovery problem using the recent framework of low-degree polynomials for recovery \cite{SCHRAMM22}, and well as providing other forms of evidence to the statistical computational gaps for the submatrix detection problem with a growing number of planted submatrices, e.g., using average-case reductions (see, for example, \cite{brennan18a}). 


\bibliographystyle{alpha}
\bibliography{bibfile}

\end{document}